\documentclass[journal,draftcls,onecolumn]{IEEEtran}
\usepackage{graphicx}
\usepackage{tabularx}
\usepackage[linesnumbered,ruled,vlined,algo2e]{algorithm2e}
\SetKw{KwDownTo}{down to}
\usepackage{booktabs}
\usepackage{threeparttable}
\usepackage{makecell}
\usepackage{multirow}
\usepackage{array}
\newcolumntype{M}[1]{>{\centering\arraybackslash}m{#1}}
\newcolumntype{C}{>{\centering\arraybackslash}X}
\usepackage[caption=false,font=normalsize,labelfont=sf,textfont=sf]{subfig}
\usepackage{textcomp}
\usepackage{stfloats}
\usepackage{url}
\usepackage{verbatim}
\usepackage{cite}
\usepackage{amsthm}
\usepackage{enumitem}
\usepackage[dvipsnames]{xcolor}
\usepackage{hyperref}
\hypersetup{
    colorlinks=true,
    urlcolor=blue,
    linkcolor=blue,
    citecolor=blue,
    bookmarksopen=true,
    bookmarksnumbered=true
}

\newtheorem{theorem}{Theorem}

\newtheorem{proposition}{Proposition}
\newtheorem{corollary}{Corollary}
\theoremstyle{definition}
\newtheorem{definition}{Definition}

\newtheorem{remark}{Remark}

\usepackage[
n,
operators,
advantage,
sets,
adversary,
landau,
probability,
notions,	
logic,
ff,
mm,
primitives,
events,
complexity,
asymptotics,
keys]{cryptocode}

\newcommand{\splitatcommas}[1]{%
	\begingroup
	\begingroup\lccode`~=`, \lowercase{\endgroup
		\edef~{\mathchar\the\mathcode`, \penalty0 \noexpand\hspace{0pt plus 1em}}%
	}\mathcode`,="8000 #1%
	\endgroup
}

\DeclareMathOperator{\dft}{DFT} %extra operator for writting DFT as in math mode
\DeclareMathOperator{\fft}{AFFT} %extra operator for writting DFT as in math mode

\begin{document}

\title{On the Additive FFT Techniques over Binary Extension Fields}

 \author{Susanta Samanta,~Mohammadtaghi Badakhshan,~and~Guang Gong%
 \thanks{S. Samanta and G. Gong are with the Department of Electrical and Computer Engineering, University of Waterloo, Waterloo, ON N2L 3G1, Canada (e-mail: ssamanta@uwaterloo.ca; ggong@uwaterloo.ca).}%
 \thanks{M. Badakhshan is with LayerZero Labs (e-mail: mtbadakhshan@gmail.com).}%
 }

\maketitle
\thispagestyle{plain}
\pagestyle{plain}

\begin{abstract}
    The Bailey's four-step FFT algorithm (1989) organizes a large Fourier transform into independent column FFTs, a diagonal twiddle-factor multiplication, and independent row FFTs. This decomposition can expose parallelism, improve memory locality, and support hierarchical implementations. Motivated by Bailey's matrix formulation, we develop additive FFT techniques for polynomial evaluation over affine subspaces of binary extension fields. Our key insight is that the Taylor expansion with respect to vanishing polynomials of subspaces provides a structural counterpart to Bailey's matrix formulation. It decomposes an additive FFT (AFFT) into independent sub-AFFTs associated with the columns and rows of a matrix. Based on this framework, we first present a general-basis AFFT that applies to any ordered basis and any split of the dimension, providing a unified baseline for measuring the gains from specialization. We then specialize the framework to the Cantor special basis and obtain two AFFT algorithms. The first supports an arbitrary decomposition of the AFFT dimension and exploits the Cantor special basis structure to perform the Taylor expansion stage without finite field multiplications. The second uses a decomposition that preserves the binomial form of the relevant subspace polynomials. It requires exactly $\frac{1}{2}n\log_2 n$ multiplications, together with a closed-form addition count determined by the binary representation of $m$. Our implementation results show that this algorithm is faster than the LCH AFFT over a Cantor special basis in 37 of the 42 configurations tested across two hardware platforms. This performance advantage stems from its fully recursive structure, which provides memory locality by design and avoids separate basis-conversion and evaluation stages. We further generalize the butterfly phase of the LCH additive FFT in the novel polynomial basis to support an arbitrary decomposition of the AFFT dimension, while retaining $n\log_2 n$ additions and $\frac{1}{2}n\log_2 n$ multiplications independently of the split. Finally, in a separate analysis, we formalize the notion of \emph{partial Cantor special bases} and identify parameter regimes in which both the von zur Gathen--Gerhard algorithm and our general-basis AFFT require fewer additions and multiplications than the first Gao--Mateer algorithm. Notably, the range of parameters for which our general-basis AFFT provides this improvement is substantially wider than that for the von zur Gathen--Gerhard algorithm.
\end{abstract}

\begin{IEEEkeywords}
Polynomial evaluation, FFT, binary extension field, affine subspace, vanishing polynomial, additive FFT, arithmetic complexity, Taylor expansion.
\end{IEEEkeywords}

\section{Introduction}\label{Sec:Introduction}
Evaluating a polynomial of degree less than $n$ at $n$ distinct points, commonly called \emph{multipoint evaluation}, is a central computational primitive across coding theory, cryptography, and signal processing. In the classical discrete Fourier transform (DFT) over a field that contains a primitive $n$th root of unity $\omega$, the evaluation domain consists of the consecutive powers $1, \omega, \ldots, \omega^{n-1}$. When $n$ is a smooth integer, the Cooley--Tukey fast Fourier transform (FFT)~\cite{CooleyTukey1965} carries out this evaluation in $O(n \log n)$ arithmetic operations by leveraging the multiplicative structure inherent in the roots of unity. From an algebraic perspective, the divide-and-conquer structure of the Cooley--Tukey algorithm corresponds to iterated factorizations of the polynomial $x^n-1$. Specifically, given a decomposition $n=n_1n_2$, one writes $x^n-1 = \prod_{j=0}^{n_2-1} \left(x^{n_1}-\omega^{j n_1}\right)$. Under this interpretation, the input polynomial is reduced modulo each factor, and the resulting remainders are then recursively reduced modulo further subfactors until only linear terms of the form $x-\omega^i$ remain. Because the remainder of a polynomial $f(x)$ upon division by $x-\omega^i$ equals the evaluation $f(\omega^i)$, the final set of residues yields exactly the desired point evaluations.

Over a binary extension field $\mathbb{F}_{2^k}$, however, the evaluation points of interest form an affine subspace of $\mathbb{F}_{2^k}$ rather than a multiplicative subgroup. Since the multiplicative group $\mathbb{F}_{2^k}^{\star}$ has odd order $2^k-1$, it contains no nontrivial elements of even order. In particular, it cannot provide the roots of unity required for a classical DFT of length $n=2^m$. Additive fast Fourier transforms (AFFTs) overcome this obstruction by exploiting the additive structure of the evaluation domain. Let $W_m = \langle \beta_0,\beta_1,\ldots,\beta_{m-1}\rangle \subseteq \mathbb{F}_{2^k}$ be an $m$-dimensional subspace of $\mathbb{F}_{2^k}$ and let $\theta \in \mathbb{F}_{2^k}$. Given a polynomial $f(x)\in \mathbb{F}_{2^k}[x]$ of degree less than $n=2^m$, the corresponding additive DFT evaluates $f$ at all points of the affine subspace $\theta+W_m$. An AFFT computes these evaluations efficiently by exploiting the additive structure of $W_m$.

Additive FFTs can be used in zero-knowledge succinct non-interactive arguments of knowledge (zkSNARKs) over binary extension fields. A zkSNARK enables a prover to convince a verifier that it knows a witness $w$ satisfying a statement $x$, using a proof that is succinct, efficiently verifiable, non-interactive, and zero-knowledge; that is, the proof reveals no information about $w$ beyond the validity of the statement $x$. Such protocols typically reduce the computation to polynomial relations and require the prover to interpolate and evaluate the resulting polynomials over large affine subspaces. These operations are precisely the tasks addressed by additive FFTs and can account for a substantial portion of the prover's running time. Examples of such protocols include Ligero~\cite{Ames2017Ligero}, STARK~\cite{STARK2018}, Aurora~\cite{Aurora2019}, Fractal~\cite{Chiesa2020Fractal}, and Polaris~\cite{Polaris2022}. They also enable fast encoding and decoding of Reed--Solomon codes~\cite{Gao2003,LCH-basis2014,LCH-FFT2016,LCH-conv2016} and fast multiplication of binary polynomials~\cite{LCH-Fast_Mult2018,LCH-Frobenius2018_2,LCH-Frobenius2018}. The performance of these applications depends not only on the exact numbers of additions and multiplications, but also on the choice of polynomial basis and the overhead incurred by basis conversions. For example, in~\cite{BSG2026}, the authors show that, when integrated into the Aurora prover, their optimized Cantor algorithm outperforms the LCH algorithm for smaller inputs. The slower performance of the LCH-based implementation is likely attributable to the memory-access overhead associated with basis conversion. This comparison illustrates that the interaction between an FFT algorithm and the underlying protocol structure can substantially affect overall system efficiency, beyond what standalone benchmarks capture.

The concept of the AFFT emerged in the late 1980s. It was first introduced by Wang and Zhu in 1988~\cite{WangZhu1988}, followed independently by Cantor in 1989~\cite{Cantor1989FFT}. These methods focus on evaluating polynomials at the roots of the vanishing polynomial associated with a given subspace or affine subspace. In his work~\cite{Cantor1989FFT}, Cantor proposed an FFT algorithm for evaluating a polynomial $f(x)$ of degree less than $n = 2^m$ over an $m$-dimensional affine subspace of the field (or subfield) $\mathbb{F}_{2^k}$, where $k = 2^\ell$ for some $\ell$. Later, von zur Gathen and Gerhard~\cite{zurGathenFFT} extended this approach to support any arbitrary $k$, though this generalization resulted in increased computational complexity. Gao and Mateer~\cite{Gao2010FFT} later introduced two alternative additive FFT algorithms based on Taylor expansion. Their first algorithm, applicable to arbitrary values of $k$, demonstrated improved efficiency compared to the method proposed by von zur Gathen and Gerhard. The second algorithm, designed for AFFT of length $2^m$ (where $m$ is a power of 2), maintained the same number of multiplications as Cantor's approach but optimized the number of additions.

In 2014, Lin et al.~\cite{LCH-basis2014} introduced a novel polynomial basis constructed using subspace polynomials over $\mathbb{F}_{2^k}$ for FFT computations. Their evaluation method processes polynomial coefficients in the LCH basis, achieving a complexity of $O(n \log n)$ for both additions and multiplications. The algorithmic advantages of this novel polynomial basis were initially showcased through fast encoding and decoding algorithms for Reed--Solomon codes, which were subsequently refined in~\cite{LCH-FFT2016,LCH-conv2016}. Further in~\cite{LCH-conv2016}, Lin et al. tackled the problem of converting between the LCH and monomial bases. When the subspace $W_m$ is generated by a Cantor special basis, their algorithms perform this conversion using $O(n\log n\log\log n)$ additions and no multiplications, with further refinements presented in~\cite{COXON2021}. These methods have been applied to fast binary polynomial multiplication~\cite{LCH-Fast_Mult2018,LCH-Frobenius2018_2,LCH-Frobenius2018}. The Cantor special basis also plays a key role in optimizing other AFFTs. For instance, the second FFT algorithm by Gao and Mateer~\cite{Gao2010FFT}, designed for lengths $n = 2^m$ with $m = 2^{t}$ for some $t$, requires only $\frac{1}{2} n \log n$ multiplications. Moreover, when the first
Gao--Mateer algorithm is implemented using a Cantor basis, the polynomial-scaling operation required in the general basis setting can be avoided at each recursive step~\cite{BernsteinChou2014,BSG2026}.

The lack of sufficiently smooth roots of unity has also motivated FFT-like algorithms over alternative algebraic evaluation domains. The elliptic-curve FFT (ECFFT)~\cite{ECFFT1_2023} constructs evaluation sets from suitable cosets of elliptic-curve subgroups and uses isogenies of smooth degree to obtain the recursive maps underlying the transform. The circle FFT~\cite{CircleFFT2024} instead evaluates functions on the circle curve $x^2+y^2=1$ over suitable finite fields. More generally, the G-FFT of Li and Xing~\cite{GFFT2024} derives FFT algorithms from automorphism subgroups of the rational function field $\mathbb{F}_q(x)$. Its affine-group instances recover the multiplicative and additive FFT settings, corresponding to smooth transform lengths dividing $q-1$ and $q$, respectively, while non-affine cyclic subgroups yield transforms for smooth lengths dividing $q+1$. These constructions enlarge the range of algebraic domains that support FFT-like computations. Their fastest bounds rely on polynomial bases adapted to the relevant subgroup or recursive algebraic
structure. In particular, the affine G-FFT applied directly to polynomials in the standard monomial basis has complexity $O(n\log^2 n)$. In this work, we restrict our attention to the direct evaluation of polynomials represented in the standard monomial basis over affine
$\mathbb{F}_2$-subspaces of binary extension fields.

\paragraph*{\textbf{Our Contributions}}

This paper develops an additive counterpart of the Bailey's four-step FFT~\cite{Bailey1990} for polynomial evaluation. The central observation is that a Taylor expansion with respect to a subspace vanishing polynomial provides the algebraic decomposition needed to organize an AFFT into independent column and row transforms. More precisely, for a decomposition $m=m_1+m_2$, the Taylor coefficients are arranged in a $2^{m_2}\times 2^{m_1}$ matrix, after which the evaluation is completed by independent column AFFTs followed by independent row AFFTs. 

The resulting matrix decomposition provides a unified framework for the algorithms developed in this paper. At the general-basis level, it supports an arbitrary dimension split and thereby permits the dimensions of the column and row sub-AFFTs to be selected according to the computational architecture. Its specialization to a Cantor special basis leads to Algorithms~\ref{Algo:AFFT-Cantor-any-m1} and~\ref{Algo:AFFT-Cantor-fixed-m1}. 
% Algorithm~\ref{Algo:AFFT-Cantor-any-m1} retains the freedom to choose any split, whereas Algorithm~\ref{Algo:AFFT-Cantor-fixed-m1} uses a prescribed recursive split that preserves the binomial form of the relevant subspace vanishing polynomial and consequently achieves a lower addition count. Both algorithms require exactly $\frac{1}{2}n\log_2 n$ multiplications. 
The arithmetic costs of the proposed algorithms and the existing Gao--Mateer~\cite{Gao2010FFT}, LCH~\cite{LCH-basis2014,LCH-FFT2016}, and Cantor AFFTs~\cite{Cantor1989FFT} are summarized in Table~\ref{Table:AFFT_cost-cmp}. 
% The exact operation counts for the Gao–Mateer and Cantor rows of Table~\ref{Table:AFFT_cost-cmp} are established in the work~\cite{BSG2026}; the entries proved in the present paper are exactly those for Algorithms~\ref{Algo:general-basis} and~\ref{Algo:AFFT-Cantor-fixed-m1}, via Theorems~\ref{Th:costAlgo1} and \ref{Theorem:cost-Algo-fixed-m1} respectively. 
In addition, the freedom provided by the general framework allows Algorithm~\ref{Algo:general-basis} to benefit from an available partial Cantor special basis, yielding parameter regimes in which it outperforms the first Gao--Mateer algorithm. The main contributions are as follows.

% which have the same multiplication count but offer different tradeoffs between splitting flexibility and addition cost. In addition, the freedom provided by the general framework allows Algorithm~\ref{Algo:general-basis} to benefit from an available partial Cantor special basis, yielding parameter regimes in which it outperforms the first Gao--Mateer algorithm~\cite{Gao2010FFT}. ({\color{blue} refer Table I}) The main contributions are as follows.

\begin{enumerate}
    \item \textbf{A general-basis AFFT with an arbitrary dimension split.} We present an AFFT (Algorithm~\ref{Algo:general-basis}), an additive counterpart of the Bailey's four-step FFT, for polynomial evaluation over an affine subspace generated by any ordered basis. For any decomposition $m=m_1+m_2$, the algorithm first computes the Taylor expansion of the input polynomial with respect to the subspace vanishing polynomial $Z_{W_{m_1}}(x)$, arranges the resulting coefficients in a matrix, and then performs independent column and row AFFTs. We prove that, for $n=2^m$, the algorithm requires
    \(
    \frac{1}{4}n(\log_2 n)^2 +\frac{3}{4}n\log_2 n
    \)
    additions and the same number of multiplications, independently of the choice of $m_1$ and $m_2$. In the general-basis setting, its addition count matches that of the first Gao--Mateer algorithm~\cite{Gao2010FFT}, although its multiplication count is asymptotically higher by a factor of $\mathcal{O}(\log n)$. Nevertheless, the Bailey-style matrix decomposition provides Algorithm~\ref{Algo:general-basis} several structural and practical advantages. In particular, its split-invariant complexity permits the dimensions of the column and row sub-AFFTs to be selected according to implementation considerations, including parallelism, cache utilization, and memory organization, without changing the arithmetic cost. Moreover, its uniform matrix-decomposition structure serves as the common foundation for the specializations to Cantor special basis  (Algorithms~\ref{Algo:AFFT-Cantor-any-m1} and~\ref{Algo:AFFT-Cantor-fixed-m1}), whose operation counts are competitive with those of the best known AFFTs (see Section~\ref{Sec:AFFT-Cantor}).

    \item \textbf{A specialization for a Cantor special basis.} We specialize the matrix-decomposition framework to affine subspaces generated by a Cantor special basis. In this setting, the coefficients of the relevant subspace vanishing polynomials lie in $\mathbb{F}_2$, so the Taylor expansion stage requires no finite field multiplications. Algorithm~\ref{Algo:AFFT-Cantor-any-m1} allows any split $m=m_1+m_2$, whereas Algorithm~\ref{Algo:AFFT-Cantor-fixed-m1} recursively chooses $m_1$ as the largest power of two smaller than $m$ to preserve the binomial form of the vanishing polynomial and reduce additions. Both algorithms requires exactly 
    \( \frac{1}{2}n\log_2 n \)
    multiplications, they differ in splitting flexibility and addition cost. For Algorithm~\ref{Algo:AFFT-Cantor-fixed-m1}, Theorem~\ref{Theorem:cost-Algo-fixed-m1} gives a closed-form addition count determined by the binary representation of $m$; when $m$ is a power of two, it becomes
    \(
    n\log_2 n +\frac{1}{4}n\log_2 n\log_2\log_2 n.
    \)
    Unlike the LCH AFFT~\cite{LCH-FFT2016}, this algorithm evaluates a polynomial given directly in the standard monomial basis, avoiding the memory-access overhead of the basis conversion. Its fully recursive structure, together with the uniform Taylor-expansion pattern shared by subproblems, also provides memory locality and predictable memory-access patterns by design, thereby reducing memory-access overhead. Our implementation results (Table~\ref{Table:benchmark-results}) show that the algorithm outperforms the LCH AFFT over a Cantor special basis in 37 of the 42 configurations tested across two hardware platforms. Moreover, it is faster than the LCH AFFT for every tested dimension on one of the two platforms.

    \item \textbf{An analysis of AFFTs over partial Cantor special bases.} We formalize the notion of a partial Cantor special basis, in which only a prefix of dimension $\ell \le m$ of the ordered basis satisfies the Cantor recursion. We show that this partial structure reduces the addition and multiplication counts of both the von zur Gathen--Gerhard AFFT~\cite{zurGathenFFT} and our proposed Algorithm~\ref{Algo:general-basis}. By contrast, the first Gao--Mateer algorithm~\cite{Gao2010FFT} cannot benefit from the same advantage because its evaluation basis changes during the recursion. We derive the corresponding operation counts and find that the two resulting regimes of very different width. The requirement $\ell \ge m-2$ of Theorem~\ref{Th:cost-mult-GG-partial} leaves the von zur Gathen and Gerhard AFFT with only two admissible values of $m$ in a given field, whereas Theorems~\ref{Th:cost-mult-algo1-partial} and~\ref{Th:add-cost-Algo1-partial} keep Algorithm~\ref{Algo:general-basis}, with $m_1=\ell$, ahead of the first Gao--Mateer AFFT over substantially wider ranges. For example, over $\mathbb{F}_{2^{48}}$ with $m_1=16$, it outperforms the first Gao--Mateer AFFT for all nine dimensions $17 \le m \le 25$ (Table~\ref{Tab:Algo-1_partial-cantor}).

    \item \textbf{A generalized butterfly phase for the LCH AFFT.} We formulate the butterfly phase of the LCH AFFT~\cite{LCH-FFT2016} in the novel polynomial basis for an arbitrary decomposition $m=m_1+m_2$. We prove in Theorem~\ref{Th:costGenLCH} that for every such decomposition, the generalized butterfly phase requires exactly $n\log_2 n$ additions and $\frac{1}{2}n\log_2 n$ multiplications. Thus, its arithmetic complexity is independent of the chosen decomposition.
\end{enumerate}

\paragraph*{\textbf{Organization of the Paper}}
The remainder of this paper is organized as follows. Section~\ref{Sec:Preliminaries} reviews the algebraic preliminaries required for the proposed AFFT techniques. Section~\ref{Sec:Algo-general} presents an AFFT over a general ordered basis using the Taylor expansion-based matrix decomposition and analyzes its arithmetic complexity. Section~\ref{Sec:AFFT-Cantor} specializes this framework to a Cantor special basis and derives the corresponding operation counts. Section~\ref{Sec:LCH-butterfly} generalizes the butterfly phase of the LCH AFFT in the novel polynomial basis to an arbitrary decomposition of the dimension. Section~\ref{Sec:partial-Cantor} studies AFFTs over partial Cantor special bases and identifies the parameter regimes in which the von zur Gathen--Gerhard algorithm and our general basis AFFT outperform the first Gao--Mateer algorithm. Section~\ref{Sec:Implementation} describes our implementation and reports benchmark results, including a performance comparison with the LCH AFFT. Finally, Section~\ref{Sec:conclusion} concludes the paper.

\begin{table*}[ht]
\caption{Comparison of finite field addition and multiplication counts for additive FFT algorithms of length $n=2^m$ over general and Cantor special bases. The total cost includes both basis conversion and evaluation.}
\label{Table:AFFT_cost-cmp}
\centering
\scriptsize
\setlength{\tabcolsep}{2pt}
\renewcommand{\arraystretch}{1.30}

\begin{threeparttable}
\begin{tabular}{
    @{}
    >{\raggedright\arraybackslash}p{1.5cm}
    >{\centering\arraybackslash}p{0.80cm}
    >{\centering\arraybackslash}p{0.48cm}
    >{\raggedright\arraybackslash}p{3.05cm}
    @{\hspace{8pt}}
    >{\raggedright\arraybackslash}p{4.65cm}
    @{\hspace{8pt}}
    >{\raggedright\arraybackslash}p{4.65cm}
    @{}
}
\toprule
Algorithm
& Basis
& Op.
& Basis conversion
& Evaluation
& Total cost\tnote{\ddag}
\\
\midrule

% Gao--Mateer: general basis
\multirow{4}{*}{%
  \shortstack[l]{Gao--Mateer\\\cite{BSG2026}}}
& \multirow{2}{*}{General}
& \#A
& $\displaystyle
   \frac{1}{4}n(\log_2 n)^2
   -\frac{1}{4}n\log_2 n$
& $\displaystyle n\log_2 n$
& $\displaystyle
   \frac{1}{4}n(\log_2 n)^2
   +\frac{3}{4}n\log_2 n$
\\

&
& \#M
& $\displaystyle n\log_2 n-n+1$
& $\displaystyle \frac{1}{2}n\log_2 n$
& $\displaystyle
   \frac{3}{2}n\log_2 n-n+1$
\\

\cmidrule(lr){2-6}

% Gao--Mateer: Cantor basis
& \multirow{2}{*}{Cantor}
& \#A
& $\displaystyle
   \frac{1}{4}n(\log_2 n)^2
   -\frac{1}{4}n\log_2 n$
& $\displaystyle n\log_2 n$
& $\displaystyle
   \frac{1}{4}n(\log_2 n)^2
   +\frac{3}{4}n\log_2 n$
\\

&
& \#M
& $0$
& $\displaystyle \frac{1}{2}n\log_2 n$
& $\displaystyle \frac{1}{2}n\log_2 n$
\\

\midrule

% LCH: general basis
\multirow{4}{*}{%
  \shortstack[l]{LCH\\\cite{LCH-conv2016}}}
& \multirow{2}{*}{General}
& \#A
& $\displaystyle O\!\left(n(\log_2 n)^2\right)$
& $\displaystyle n\log_2 n$
& $\displaystyle
   n\log_2 n
   +O\!\left(n(\log_2 n)^2\right)$
\\

&
& \#M
& $\displaystyle O(n\log_2 n)$
& $\displaystyle \frac{1}{2}n\log_2 n$
& $\displaystyle
   \frac{1}{2}n\log_2 n
   +O(n\log_2 n)$
\\

\cmidrule(lr){2-6}

% LCH: Cantor basis
& \multirow{2}{*}{Cantor}
& \#A
& $\displaystyle
   O\!\left(
      n\log_2 n\log_2\log_2 n
   \right)$
& $\displaystyle n\log_2 n$
& $\displaystyle
   n\log_2 n
   +O\!\left(
      n\log_2 n\log_2\log_2 n
   \right)$
\\

&
& \#M
& $0$
& $\displaystyle \frac{1}{2}n\log_2 n$
& $\displaystyle \frac{1}{2}n\log_2 n$
\\

\midrule

% Cantor algorithm
\multirow{2}{*}{%
  \shortstack[l]{Cantor\\\cite{BSG2026}}}
& \multirow{2}{*}{Cantor}
& \#A
& \multirow{2}{*}{N/A}
& $\displaystyle
   \frac{1}{2}n\log_2 n
   +\frac{1}{2}n
    \sum_{r=0}^{\log_2 n-1}
       2^{\mathrm{wt}(r)}$
& \multirow{2}{*}{\shortstack{Same as evaluation}}
\\

&
& \#M
&
& $\displaystyle \frac{1}{2}n\log_2 n$
& 
\\

\midrule

% Proposed general-basis AFFT
\multirow{2}{*}{%
  \shortstack[l]{Proposed \\
  (Algorithm~\ref{Algo:general-basis})}}
& \multirow{2}{*}{General}
& \#A
& \multirow{2}{*}{N/A}
& $\displaystyle
   \frac{1}{4}n(\log_2 n)^2
   +\frac{3}{4}n\log_2 n$
& \multirow{2}{*}{\shortstack{Same as evaluation}}
\\[8pt]

&
& \#M
&
&  $\displaystyle
   \frac{1}{4}n(\log_2 n)^2
   +\frac{3}{4}n\log_2 n$
& 
\\

\midrule

% Proposed Cantor-basis AFFT
\multirow{2}{*}{%
  \shortstack[l]{Proposed\\
  (Algorithm~\ref{Algo:AFFT-Cantor-fixed-m1})}}
& \multirow{2}{*}{Cantor}
& \#A
& \multirow{2}{*}{N/A}
& $\displaystyle
   n\log_2 n
   +n\!\left[
      \sum_{i=1}^{w}p_i2^{p_i-2}
      +\frac{1}{2}
       \sum_{i=1}^{w}
          \left(m\bmod 2^{p_i}\right)
   \right]$\tnote{\dag}
& \multirow{2}{*}{\shortstack{Same as evaluation}}
\\\addlinespace[8pt]

&
& \#M
& 
& $\displaystyle \frac{1}{2}n\log_2 n$
& 
\\

\bottomrule
\end{tabular}

\begin{tablenotes}
\footnotesize

\item[\dag]
Here, $m=2^{p_1}+2^{p_2}+\cdots+2^{p_w}$, where $p_1>p_2>\cdots>p_w\geq 0$ and $w=\mathrm{wt}(m)$. When $m=2^t$, the addition count simplifies to $ n\log_2 n +\frac{1}{4}n\log_2 n\,\log_2\log_2 n$.

\item[\ddag]
The total cost reports additions and multiplications separately and includes both basis conversion and evaluation. For algorithms that operate directly on the input polynomial representation, no separate basis-conversion stage is required.

\item[]
The exact operation counts for the Gao--Mateer~\cite{Gao2010FFT} and Cantor~\cite{Cantor1989FFT} algorithms were established in~\cite{BSG2026}. 
% The entries proved in the present paper are those for Algorithms~\ref{Algo:general-basis} and~\ref{Algo:AFFT-Cantor-fixed-m1}, as established in Theorems~\ref{Th:costAlgo1} and~\ref{Theorem:cost-Algo-fixed-m1}, respectively.

\end{tablenotes}
\end{threeparttable}
\end{table*}

\section{Preliminaries}\label{Sec:Preliminaries}
In this section, we collect the algebraic tools required by our algorithms. Let $\mathbb{F}_{2^k}$ be the finite field of order $2^k$. We know that there exists a vector space isomorphism from $\mathbb{F}_{2^k}$ to $\mathbb{F}_2^k$ defined by $\mathbf{x}=(x_0\beta_0+x_1\beta_1+ \cdots +x_{k-1}\beta_{k-1}) \mapsto (x_0,x_1, \ldots,x_{k-1})$, where $\set{\beta_0,\beta_1,\ldots,\beta_{k-1}}$ is a basis of $\mathbb{F}_{2^k}$. The polynomial ring over $\mathbb{F}_{2^k}$ in the variable $x$ is denoted by $\mathbb{F}_{2^k}[x]$. 

An \emph{$m$-dimensional subspace} of $\mathbb{F}_{2^k}$, equivalently an $m$-dimensional $\mathbb{F}_2$-vector space is a subset
\[
W_m = \langle \beta_0, \ldots, \beta_{m-1}\rangle = \bigl\{\sum_{i=0}^{m-1} c_i \beta_i : c_i \in \set{0,1} \bigr\}
\]
We order the elements of the subspace $W_m$ by $\{\eta_0=0, \eta_1, \eta_2, \dots, \eta_{2^m-1}\}$ where
\[
    \eta_j = \sum_{i=0}^{m-1} x_i \beta_i \quad \text{and} \quad j=\sum_{i=0}^{m-1} x_i 2^i, x_i\in \set{0,1}.
\]
Thus, we have $|W_m|=2^m$. For $\theta \in \mathbb{F}_{2^k}$, we call the set $\theta + W = \{\theta + w : w \in W\}$ a \emph{$m$-dimensional affine subspace} of $\mathbb{F}_{2^k}$. Two affine subspace either identical or disjoint and if $m < k-1$ and $\theta$ is any linear combination of $\{\beta_{m+1}, \beta_{m+2}, \ldots, \beta_{k-1}\}$, then we can decompose the affine subspace $\theta+W_{m+1}$ into two disjoint affine subspace as
\[
\theta+W_{m+1} = (\theta+W_m) \cup (\theta+\beta_m + W_m).
\]

\begin{definition}[Vanishing polynomial]\label{Def:vanishing}
    Let $W_m$ be a $m$-dimensional subspace of $\mathbb{F}_{2^k}$. The \emph{vanishing polynomial} of $W_m$ is given by $Z_{W_m}(x) = \prod_{\omega \in W_m} (x + \omega)$.
\end{definition}
The vanishing polynomial of an $m$-dimensional subspace is a monic $\mathbb{F}_2$-linearized polynomial. Hence, it can be written as
\[
Z_{W_m}(x) = x^{2^m} + c_{m-1}x^{2^{m-1}} + \cdots + c_1 x^2 + c_0 x,
\]
where $c_0,\ldots,c_{m-1} \in \mathbb{F}_{2^k}$. In particular, $Z_{W_m}(x)$ has degree $2^m$ and contains at most $m+1$ nonzero terms. Since $W_{m+1}= W_m \cup (\beta_{m}+W_m)$, we have 
\[
\mathbb{Z}_{W_{m+1}}(x)= (\mathbb{Z}_{W_{m}}(x))^2- \mathbb{Z}_{W_{m}}(\beta_{m})\cdot \mathbb{Z}_{W_{m}}(x).
\]

\begin{definition}[Taylor expansion]\label{Def:TaylorExp}
Let $h(x) \in \mathbb{F}_{2^k}[x]$ be a monic polynomial of degree $d$, and let $f(x) \in \mathbb{F}_{2^k}[x]$ have degree less than $nd$ for some positive integer $n$. The \emph{Taylor expansion of $f(x)$ with respect to $h(x)$} is the unique representation
\[
f(x) = \sum_{i=0}^{n-1} f_i(x) \cdot (h(x))^i,
\]
where each $f_i(x) \in \mathbb{F}_{2^k}[x]$ has $\deg(f_i) < d$.
\end{definition}

Now we will discuss the evaluation of a univariate polynomial $f(x)$ over the subspace $W_m$.

\begin{definition}[Additive Discrete Fourier Transform]
    The evaluation of $f(x)$ at the points $\eta_0, \eta_1, \ldots, \eta_{2^m-1}$ is given by $ \hat{f} = \left( f(\eta_0), f(\eta_1), \ldots, f(\eta_{2^m-1}) \right)$. This set of evaluations is referred to as the additive discrete Fourier transform (DFT) of $f(x)$ over the subspace $W_m$.
\end{definition}
We sometimes refer to the vector $\hat{f}$ as the \textit{discrete Fourier transform of length $n=2^m$} for the function $f(x)$, denoted by $\dft(f, W_m)$. \textit{The additive FFT} (AFFT) is an efficient method for computing $\dft(f, W_m)$, which we will denote as $\fft(f, W_m)$.

\noindent Consider the function $S:\mathbb{F}_{2^k} \rightarrow \mathbb{F}_{2^k}$ defined by $S(x)=x^2+x$, and let the following sequence of functions be defined recursively:
\[S^{0}(x)=x \quad \text{ and } \quad S^{m}(x)=S(S^{m-1}(x)).\]
A nonrecursive formula for $S^m(x)$ is 
\[
S^{m}(x)=\sum_{i=0}^{m} \binom{m}{i} x^{2^i},
\]
where $\binom{m}{i}$ denotes the binomial coefficient reduced modulo 2. For $m=2^{t}$, we have $S^{2^{t}}=x^{2^{2^{t}}}+x$.

\begin{definition}[Cantor Special Basis]\label{Def:Cantor-basis}
    The \emph{Cantor special basis} of the $m$ dimensional subspace $W_{m}$ of $\mathbb{F}_{2^k}$ is the ordered set $\set{ \beta_{0},\beta_{1},\ldots,\beta_{m-1}}$ satisfying
    \[
    S(\beta_{i})=\beta_{i-1} \quad \text{for} \quad i=1,\ldots,m-1 \quad \text{with} \quad \beta_{0}=1.
    \]
\end{definition}

A natural question is for which fields and dimensions such a basis exists. The following proposition gives a precise existential criterion.

\begin{proposition}\cite[Appendix]{Gao2010FFT}\label{Prop:partial-cantor}
    The field $\mathbb{F}_{2^k}$ admits a Cantor special basis of dimension $\ell$ if and only if $2^{\lceil \log_2 \ell \rceil}$ divides $k$.
\end{proposition}

Under a Cantor special basis, the vanishing polynomials take a particularly simple form. More specifically, if $\{\beta_0 = 1, \beta_1, \ldots, \beta_{m-1}\}$ is a Cantor special basis, then the vanishing polynomial of $W_i = \langle \beta_0, \ldots, \beta_{i-1}\rangle$ is
\[
Z_{W_i}(x) = S^i(x),
\]
i.e. the $i$-fold composition of the mapping $S$. In particular, $Z_{W_i}(x)$ is linearized with all coefficients in $\mathbb{F}_{2}$, so polynomial division by $Z_{W_i}$ requires \emph{zero} $\mathbb{F}_{2^k}$-multiplications.

\section{A New Additive FFT over General Basis}\label{Sec:Algo-general}
In this section, we present our first main algorithm for the additive fast Fourier transform (AFFT) over an arbitrary basis of $\mathbb{F}_{2^k}$. The algorithm recursively reduces multipoint evaluation to smaller evaluation problems by applying a Taylor expansion with respect to subspace vanishing polynomials.

\noindent Let
\[
W_m = \langle \beta_0, \beta_1, \ldots, \beta_{m-1} \rangle,
\]
and let $f(x) \in \mathbb{F}_{2^k}[x]$ be a polynomial of degree less than
\(
n = 2^m.
\)
Our goal is to evaluate $f$ over the affine space $\theta + W_m$. Given a decomposition $m = m_1 + m_2$, where $m_1, m_2 \geq 1$, the algorithm proceeds as follows:
\begin{enumerate}
    \item \textbf{Decomposition:}
    Compute the Taylor expansion of $f(x)$ with respect to the vanishing polynomial $Z_{W_{m_1}}(x)$. The resulting coefficient polynomials are arranged as a $2^{m_2} \times 2^{m_1}$ matrix.

    \item \textbf{Column evaluations:}
    Evaluate each of the $2^{m_1}$ column polynomials, each of degree less than $2^{m_2}$, over an appropriate projected affine space.

    \item \textbf{Row evaluations:}
    For each of the $2^{m_2}$ row polynomials, evaluate a polynomial of degree less than $2^{m_1}$ over the corresponding affine space $\theta_i + W_{m_1}$.
\end{enumerate}

For $m = m_1 + m_2$, we can decompose the subspace $W_m$ into $2^{m_2}$ disjoint affine subspaces 
\[
W_m = \bigcup_{i=0}^{2^{m_2}-1} (\theta_i + W_{m_1}),
\]
where $\theta_i = \sum_{j=0}^{m_2-1} c_j \beta_{m_1+j}$ with $(c_{m_2-1}, \ldots, c_0)_2$ being the binary representation of $i$. In particular, $\theta_0 = 0$. For any element $\theta_i + \omega \in \theta_i + W_{m_1}$, where $\omega \in W_{m_1}$, the $\mathbb{F}_2$-linearity of $Z_{W_{m_1}}(x)$ and the fact that $\omega \in \ker(Z_{W_{m_1}})$ imply that
\[
Z_{W_{m_1}}(\theta_i + \omega) = Z_{W_{m_1}}(\theta_i) + Z_{W_{m_1}}(\omega) = Z_{W_{m_1}}(\theta_i).
\]

Hence, $Z_{W_{m_1}}(x)$ is constant on each coset $\theta_i + W_{m_1}$, with value $Z_{W_{m_1}}(\theta_i)$. Algorithm~\ref{Algo:general-basis} describes the resulting evaluation procedure. We establish the following result to prove its correctness.

% Consider the subspace $W_m=\langle \beta_{0},\beta_{1},\ldots,\beta_{m-1} \rangle$ and we want to evaluate the polynomial $f(x)$ of degree less than $n=2^m$ over the affine space $\theta+W_m$.

\begin{theorem}\label{Th:injectivity}
Let $m=m_1+m_2$, and let
\[
W_m
=
\left\langle
\beta_0,\beta_1,\ldots,\beta_{m-1}
\right\rangle
\]
be an $m$-dimensional $\mathbb F_2$-vector space. Define
\[ 
W_{m_1}
=
\left\langle
\beta_0,\beta_1,\ldots,\beta_{m_1-1}
\right\rangle
\quad
\text{and}
\quad
U
=
\left\langle
\beta_{m_1},\beta_{m_1+1},\ldots,\beta_{m_1+m_2-1}
\right\rangle
\]
Then $W_m = W_{m_1}\oplus U$. Let $Z_{W_{m_1}}(x)$ be the vanishing polynomial of $W_{m_1}$. Then the restriction $ Z_{W_{m_1}}\big|_U:U\to Z_{W_{m_1}}(U)$ is an injective $\mathbb F_2$-linear map. Consequently, $Z_{W_{m_1}}(U)$
is an $m_2$-dimensional $\mathbb F_2$-linear subspace. Moreover, the set
\[
\left\{
Z_{W_{m_1}}(\beta_{m_1}),
Z_{W_{m_1}}(\beta_{m_1+1}),
\ldots,
Z_{W_{m_1}}(\beta_{m_1+m_2-1})
\right\}
\]
is a basis of $Z_{W_{m_1}}(U)$ over $\mathbb F_2$.
\end{theorem}

\begin{proof}
Since $W_m= \left\langle
\beta_0,\beta_1,\ldots,\beta_{m-1}
\right\rangle$, and since $m=m_1+m_2$, we have
\[
W_m
=
\left\langle
\beta_0,\ldots,\beta_{m_1-1}
\right\rangle
\oplus
\left\langle
\beta_{m_1},\ldots,\beta_{m_1+m_2-1}
\right\rangle.
\]
Hence $W_m=W_{m_1}\oplus U$, where $W_{m_1} = \left\langle \beta_0,\ldots,\beta_{m_1-1} \right\rangle$ and $U= \left\langle
\beta_{m_1},\beta_{m_1+1},\ldots,\beta_{m_1+m_2-1}
\right\rangle$. In particular, we have $U\cap W_{m_1}=\{0\}$.

Now let $Z_{W_{m_1}}(x)$ be the vanishing polynomial of $W_{m_1}$. Since $Z_{W_{m_1}}$ vanishes exactly on $W_{m_1}$, we have
$\ker(Z_{W_{m_1}})=W_{m_1}$. Moreover, $Z_{W_{m_1}}$ is an $\mathbb F_2$-linearized polynomial, so it defines an $\mathbb F_2$-linear map on $W_m$.

Restricting to $U$, we get $Z_{W_{m_1}}\big|_U:U\to Z_{W_{m_1}}(U)$. Its kernel is
\[
\ker\left(Z_{W_{m_1}}\big|_U\right)
=
U\cap \ker Z_{W_{m_1}}
=
U\cap W_{m_1}
=
\{0\}.
\]
Therefore $Z_{W_{m_1}}\big|_U$ is injective. Since $U=
\left\langle
\beta_{m_1},\beta_{m_1+1},\ldots,\beta_{m_1+m_2-1}
\right\rangle$, the vectors
\[
\beta_{m_1},\beta_{m_1+1},\ldots,\beta_{m_1+m_2-1}
\]
form a basis of $U$. Because $Z_{W_{m_1}}\big|_U$ is injective and $\mathbb F_2$-linear, their images
\[
Z_{W_{m_1}}(\beta_{m_1}),
Z_{W_{m_1}}(\beta_{m_1+1}),
\ldots,
Z_{W_{m_1}}(\beta_{m_1+m_2-1})
\]
are linearly independent in $Z_{W_{m_1}}(U)$. Therefore, they form a basis of $Z_{W_{m_1}}(U)$. Hence $\dim Z_{W_{m_1}}(U)=\dim(U)=m_2$.

\end{proof}

\noindent Theorem~\ref{Th:injectivity} shows that the restriction of $Z_{W_{m_1}}$ to $U$ is an $\mathbb{F}_2$-linear isomorphism onto its image. Consequently, 
\[
Z_{W_{m_1}}(U)
=
\left\langle
Z_{W_{m_1}}(\beta_{m_1}),\ldots,
Z_{W_{m_1}}(\beta_{m-1})
\right\rangle
\]
is an $m_2$-dimensional subspace. Thus, the distinct cosets of $W_{m_1}$ in $W_m$ are mapped to distinct points of $Z_{W_{m_1}}(U)$. This observation allows the evaluation over $\theta+W_m$ to be organized as a row-column computation.

The Taylor expansion of $f(x)$ with respect to $Z_{W_{m_1}}(x)$ can be written as
\begin{equation}
\label{eq:Taylor-matrix-decomposition}
\begin{aligned}
f(x)
&=
\sum_{i=0}^{2^{m_2}-1}
\left(
\sum_{j=0}^{2^{m_1}-1} g_{i,j}x^j
\right) (Z_{W_{m_1}}(x))^i        \\
&=
\sum_{j=0}^{2^{m_1}-1}
\left(
\sum_{i=0}^{2^{m_2}-1} g_{i,j}Z_{W_{m_1}}(x)^i
\right)x^j.
\end{aligned}
\end{equation}

Arrange the coefficients $g_{i,j}$ in a $2^{m_2}\times 2^{m_1}$ matrix $M$ such that $M[i,j] = g_{i,j}$. 

For $0\leq j<2^{m_1}$, define the each column $j$ by
\[
C_j(y)= \sum_{i=0}^{2^{m_2}-1}M[i,j]y^i.
\]
Equation~\eqref{eq:Taylor-matrix-decomposition} then becomes
\[
f(x)=\sum_{j=0}^{2^{m_1}-1} C_j\bigl(Z_{W_{m_1}}(x)\bigr)x^j.
\]

We next relate the evaluation points of the column polynomials to the cosets of $W_{m_1}$ in $W_m$. Since $W_m=W_{m_1}\oplus U$, every element of $W_m$ can be written uniquely as $\omega+u$, where $\omega\in W_{m_1}$ and $u\in U$. Consequently,
\[
\theta+W_m=
\bigcup_{u\in U}
\bigl( \theta+u+W_{m_1} \bigr).
\]

Enumerate the elements of $U$ as
\[
u_k= \sum_{\ell=0}^{m_2-1}c_\ell\beta_{m_1+\ell},
\qquad \text{where} \quad
(c_{m_2-1},\ldots,c_0)_2=k,
\]
for $0\leq k<2^{m_2}$, and define $\theta_k=\theta+u_k$. With this notation, the affine space is partitioned into the $2^{m_2}$ disjoint cosets
\[
\theta+W_m =
\bigcup_{k=0}^{2^{m_2}-1}
\bigl(\theta_k+W_{m_1}\bigr).
\]

Since $Z_{W_{m_1}}$ is $\mathbb{F}_2$-linear,
\[
Z_{W_{m_1}}(\theta_k) = Z_{W_{m_1}}(\theta) + Z_{W_{m_1}}(u_k).
\]
It follows that the points $Z_{W_{m_1}}(\theta_k)$ form the projected affine space
\[
U'= Z_{W_{m_1}}(\theta) + Z_{W_{m_1}}(U) = Z_{W_{m_1}}(\theta)
+
\left\langle
Z_{W_{m_1}}(\beta_{m_1}),
\ldots,
Z_{W_{m_1}}(\beta_{m_1+m_2-1})
\right\rangle.
\]
The injectivity established in Theorem~\ref{Th:injectivity} ensures that these $2^{m_2}$ points are distinct.

Evaluating each column polynomial $C_j$ over $U'$ replaces the entry in row $k$ and column $j$ by
\[
M[k,j]= C_j\bigl(Z_{W_{m_1}}(\theta_k)\bigr).
\]
Consequently, the $k$-th row of the updated matrix defines the polynomial
\[
R_k(x)= 
\sum_{j=0}^{2^{m_1}-1}
C_j\bigl(Z_{W_{m_1}}(\theta_k)\bigr)x^j.
\]

For every $\omega\in W_{m_1}$, the vanishing property and linearity of $Z_{W_{m_1}}$ imply that
\[
\begin{aligned}
Z_{W_{m_1}}(\theta_k+\omega)
&=
Z_{W_{m_1}}(\theta_k) + Z_{W_{m_1}}(\omega) =
Z_{W_{m_1}}(\theta_k).
\end{aligned}
\]
Thus, it follows that
\[
\begin{aligned}
R_k(\theta_k+\omega)
&=
\sum_{j=0}^{2^{m_1}-1}
C_j\bigl(Z_{W_{m_1}}(\theta_k)\bigr)(\theta_k+\omega)^j \\
&=
\sum_{j=0}^{2^{m_1}-1} C_j\bigl(Z_{W_{m_1}}(\theta_k+\omega)\bigr)(\theta_k+\omega)^j \\
&= f(\theta_k+\omega).
\end{aligned}
\]

Therefore, evaluating $R_k$ over $\theta_k+W_{m_1}$ produces exactly the values of $f$ on that coset.
% 
% Finally, the direct-sum decomposition $W_m=W_{m_1}\oplus U$ gives
% \[
% \theta+W_m =
% \bigcup_{k=0}^{2^{m_2}-1}
% \bigl(\theta_k+W_{m_1}\bigr).
% \]
Hence, concatenating the row evaluations produces the evaluation of $f$ over the entire affine space $\theta+W_m$. Algorithm~\ref{Algo:general-basis} formalizes this column-row evaluation procedure.

% Then
% \begin{itemize}
%     \item Each \emph{column} $M[*,j]$ defines a polynomial $C_j(y) = \sum_{i=0}^{2^{m_2}-1} M[i,j]\, y^i$ to be evaluated at the points of $Z_{W_{m_1}}(\theta) + \langle Z_{W_{m_1}}(\beta_{m_1}), \ldots, Z_{W_{m_1}}(\beta_{m_1+m_2-1}) \rangle$
%     \item After column evaluations update $M$, each \emph{row} $M[i,*]$ defines a polynomial $R_i(x) = \sum_{j=0}^{2^{m_1}-1} M[i,j]\, x^j$ to be evaluated over the coset $\theta_i + W_{m_1}$.
% \end{itemize}

\begin{algorithm2e}[ht]
    \caption{Additive FFT of length $n = 2^m$ over a general basis}\label{Algo:general-basis}
    \SetAlgoLined
    {\scriptsize
    \KwIn{$f(x) \in \mathbb{F}_{2^k}[x]$ of degree $< n = 2^m$ and the affine subspace $\theta+W_m=\theta+ \langle \beta_0,\beta_1,\ldots,\beta_{m-1} \rangle$.}
    \KwOut{$\fft(f, \theta+W_m)$.}

    \If{$m=1$}
    {
        \Return $(f(\theta),f(\theta+\beta_0))$\;
    }
    Let $m = m_1 + m_2$ and $U=\left\langle \beta_{m_1},\beta_{m_1+1},\ldots,\beta_{m_1+m_2-1} \right\rangle$\;
    
    Do the Taylor Expansion of $f(x)$  w.r.t. $Z_{W_{m_1}}(x)$ 
    \begin{equation*}
        \begin{aligned}
            f(x)
            &=\sum_{i=0}^{2^{m_2}-1} \left( \sum_{j=0}^{2^{m_1}-1} g_{i,j} x^j \right) (Z_{W_{m_1}}(x))^i
        \end{aligned}
    \end{equation*}

    Initialize a $2^{m_2} \times 2^{m_1}$ matrix $M$ such that $M[k_2, k_1] = g_{k_2, k_1}$\;

    Let $U' = Z_{W_{m_1}}(\theta) + Z_{W_{m_1}}(U) = Z_{W_{m_1}}(\theta) + \langle Z_{W_{m_1}}(\beta_{m_1}), \ldots, Z_{W_{m_1}}(\beta_{m_1+m_2-1}) \rangle$\;
    
    \tcc{Compute column evaluations over the affine space $U'$}
    \For{$k_1 = 0, \dots, 2^{m_1}-1$}{
        Define $C_{k_1}(y) = \sum_{k_2=0}^{2^{m_2}-1} M[k_2, k_1] y^{k_2}$\;
        
        % Let the column polynomial be $C_{k_1}(x) = \sum_{k_2=0}^{2^{m_2}-1} M[k_2, k_1] x^{k_2}$\;
        $M[*, k_1] \leftarrow \fft(C_{k_1}(y), U')$\;
    }

    \tcc{Compute row polynomials over the affine spaces}
    \For{$k_2 = 0, \dots, 2^{m_2}-1$}{
        % Let $(c_{m_2-1}, \ldots, c_0)_2$ be the binary representation of $k_2$\;
        Let $\theta_{k_2} = \theta + \sum_{\ell=0}^{m_2-1} c_{\ell} \beta_{m_1+\ell}$, where $(c_{m_2-1}, \ldots, c_0)_2=k_2$\;
        
        Define $R_{k_2}(x) = \sum_{k_1=0}^{2^{m_1}-1} M[k_2, k_1] x^{k_1}$\;
        
        $E_{k_2} \leftarrow \fft(R_{k_2}(x), \theta_{k_2} + W_{m_1})$\;
    }

    \Return $E_0 \,||\, E_1 \,||\, \dots \,||\, E_{2^{m_2}-1}$
    }
\end{algorithm2e}

\begin{remark}
Algorithm~\ref{Algo:general-basis} is an additive analogue of the Bailey's four-step FFT~\cite{Bailey1990}. The Taylor expansion with respect to $Z_{W_{m_1}}(x)$ plays a role analogous to Bailey's matrix formulation. It arranges the coefficients of $f(x)$ into a $2^{m_2}\times 2^{m_1}$ matrix and decomposes the evaluation problem into independent column and row sub-AFFTs. The column sub-AFFTs are evaluated over the projected affine space
\(
Z_{W_{m_1}}(\theta)+Z_{W_{m_1}}(U),
\)
whereas the row sub-AFFTs are evaluated over the cosets $\theta_k+W_{m_1}$ for $0\leq k<2^{m_2}$. In contrast to Bailey's classical algorithm, no separate diagonal multiplication by twiddle factors is needed. Instead, the Taylor expansion and the projected evaluation points encode the interaction between the two stages.
\end{remark}

% $\theta_i+W_{m_1}$, where $0\le i 
% \le 2^{m_2}-1$ and $\theta_i$ is linear combination of $\set{\beta_{m_1},\beta_{m_1+1},\ldots,\beta_{m_1+m_2-1}}$ with $\theta_0=0$. 

% The algorithm is presented in Algorithm~\ref{Algo:general-basis}.

% Therefore, for each $\theta_i+\omega \in \theta_i+W_{m_1}$, we have $Z_{W_{m_1}}(\theta_i+\omega)=Z_{W_{m_1}}(\theta_i)$ which implies that for each $x\in \theta_i+W_{m_1}$, we have $Z_{W_{m_1}}(x) \in Z_{W_{m_1}}(U)$, where $Z_{W_{m_1}}(U)$ is an $m_2$-dimensional $\mathbb F_2$-linear subspace with 
% \[
% U
% =
% \left\langle
% \beta_{m_1},\beta_{m_1+1},\ldots,\beta_{m_1+m_2-1}
% \right\rangle.
% \]

% If we do the Taylor expansion of the input $f(x)$ with respect to the vanishing polynomial $Z_{W_{m_1}}(x)$, then we can write $f(x)$ as
% \[f(x) = \sum_{i=0}^{n_2-1} \left( \sum_{j=0}^{n_1-1} g_{i,j} x^j \right) (Z_{W_{m_1}}(x))^i, \]
% where $n_1=2^{m_1}$ and $n_2=2^{m_2}$. We can rewrite $f(x)$ as follows
% \[f(x) = \sum_{j=0}^{n_1-1} \left( \sum_{i=0}^{n_2-1} g_{i,j} (Z_{W_{m_1}}(x))^i \right) x^j. \]

\noindent We now turn to the arithmetic complexity of Algorithm~\ref{Algo:general-basis}. The cost decomposes naturally into two parts, namely the Taylor expansion step, which divides $f(x)$ by powers of $Z_{W_{m_1}}(x)$, and the sub-FFT evaluations in the column and row passes. We analyze each in turn and show that the total cost is invariant under the choice of split $m=m_1+m_2$.

\subsection{Detailed Cost Analysis}\label{Sec:costNewAFFT} 
We first analyze the cost of the Taylor expansion subroutine in the algorithm, where Algorithm~\ref{Algo:RecursiveTE} computes the Taylor expansion of $f(x)$ with respect to the vanishing polynomial $Z_{W_{m_1}}(x)$.

\begin{algorithm2e}[ht]
    \caption{Computing $\mathrm{TE}(f(x),h(x)=Z_{W_{m_1}}(x))$}
    \label{Algo:RecursiveTE}
    \SetAlgoLined
    {\scriptsize
    \KwIn{A polynomial $f(x)$ with $\deg(f)<2^{m_1+m_2}$, integers $m_1,m_2$, and the polynomial $h(x)=Z_{W_{m_1}}(x)$.}
    \KwOut{The Taylor Expansion representation of $f$ with respect to $h=Z_{W_{m_1}}$.}

    $e \leftarrow 2^{m_2-1}$\;

    % \If{$e<1$}{
    %     \Return $f(x)$\;
    % }

    Compute the division
    $$
        f(x)=r(x)+q(x)h^e,
        \qquad \deg(r)<\deg(h^e) \quad \text{and} \quad \deg(q)<\deg(h^e).
    $$

    \If{$m_2=1$}{
    \Return $(r,q)$
    }

    \Return $\mathrm{TE}(r,h,m_2-1)\; || \; \mathrm{TE}(q,h,m_2-1)$\;
    }
\end{algorithm2e}

We now quantify the cost of the Taylor expansion subroutine. The key observation is that every power of the vanishing polynomial $h(x)=Z_{W_{m_1}}(x)$ retains at most $m_1+1$ nonzero terms, so each division step costs only $m_1$ multiplications regardless of the recursion depth.

\begin{proposition}\label{Prop:costTE}
   Algorithm~\ref{Algo:RecursiveTE} requires $m_1m_2 2^{m-1}$ multiplications and $m_1m_22^{m-1}$ additions, where $m=m_1+m_2$.
\end{proposition}

\begin{proof}
    At the first step, the first division is by $h^{2^{m_2}-1}$ and $\deg(h^{2^{m_2}-1})=2^{m-1}$. Since $h(x)$ is the vanishing polynomial of the subspace $W_{m_1}$, both $h(x)$ and every power $h^e$ contain at most $m_1+1$ nonzero terms. Since the highest-degree term in $h^{2^{m_2}-1}$ has a coefficient of one, each iteration of dividing $f(x)$ by $h^{2^{m_2}-1}$ requires $m_1$ multiplications. Additionally, at each iteration, the highest-degree term cancels out, allowing us to omit the corresponding addition. Consequently, dividing $f(x)$ by $h^{2^{m_2}-1}$ requires a total of $2^{m-1} \cdot m_1$ multiplications and additions.

    More generally, at recursion level $i$, where $0 \leq i \leq m_2-1$, there are $2^i$ subproblems. Each subproblem involves a polynomial of degree less than $2^{m-i}$ and a divisor $h^{2^{m_2-1-i}}$, whose degree is
    \[
    \deg\!\left(h^{2^{m_2-1-i}}\right)=2^{m_1}\cdot 2^{m_2-1-i}=2^{m-1-i}.
    \]
    Hence, each subproblem requires $2^{m-1-i}m_1$ multiplications and the same number of additions. Summing over all recursion levels, the total cost is
    \[
    \sum_{i=0}^{m_2-1} 2^i \cdot 2^{m-1-i} m_1
    = m_1m_22^{m-1}.
    \]
    Therefore, Algorithm~\ref{Algo:RecursiveTE} performs $m_1m_22^{m-1}$ multiplications and $m_1m_22^{m-1}$ additions.
\end{proof}

We now combine the cost of the Taylor expansion subroutine with that of the recursive sub-FFT calls to obtain the total arithmetic complexity of Algorithm~\ref{Algo:general-basis}.

\begin{theorem}\label{Th:costAlgo1}
    Algorithm~\ref{Algo:general-basis} evaluates a polynomial of degree less than $2^m$ over an $m$-dimensional affine subspace using
    \[
    \frac{1}{4}2^m m^2 + \frac{3}{4}2^m m = \frac{1}{4}n(\log_2 n)^2 + \frac{3}{4}n\log_2 n
    \]
    multiplications and the same number of additions, regardless of the choice of split $m=m_1+m_2$.
\end{theorem}

\begin{proof}
The cost decomposes into two parts, the Taylor expansion step and the sub-FFT evaluations.

\noindent Since the numbers of additions and multiplications coincide (by Proposition~\ref{Prop:costTE}), let $A_{TE}(m)$ denote this common quantity. The algorithm produces $2^{m_1}$ column subproblems and $2^{m_2}$ row subproblems, $A_{TE}(m)$ satisfies the recurrence
\begin{equation}\label{Eqn:TE-cost}
    A_{TE}(m) = 2^{m_1} A_{TE}(m_2) + 2^{m_2} A_{TE}(m_1) + \frac{1}{2} 2^m m_1 m_2.
\end{equation}

Define $\tau(r)=A_{TE}(r)/2^r$. Dividing both sides by $2^{m}$, we obtain
\[\tau(m)= \tau(m_1)+ \tau(m_2)+ \frac{1}{2}m_1m_2. \]

% If we split $m=m_1+m_2$ with $m_1=m_2=\frac{m}{2}$, then the recursion simply reduces to
% \[\tau(m) = 2 \tau(m/2) +  \frac{1}{8} {m^2} \]
% with $\tau(0)=0$. This gives
% \[
% \tau(m)= \frac{1}{4}(m^2-m).
% \]

% However, the parameter $m$ may be split into any pair of integers $m_1,m_2 \geq 1$ such that $m_1+m_2=m$. The overall complexity remains unchanged. We prove this by strong induction on $m$. For the base case, $\tau(1)=0$. So the statement is true for $m=1$. Assume that for some $m\geq 2$, the formula holds for all positive integers $k<m$. That is, assume $\tau(k)=\frac{1}{4}(k^2-k)$ for all $1\leq k < m$. We must prove that $\tau(m)=\frac{1}{4}(m^2-m)$.

We show $\tau(m)=\frac{1}{4}(m^2-m)$ for all $m\geq 1$ by strong induction, independent of the split. The base case $\tau(1)=0=\frac{1}{4}(1-1)$ is immediate. Assume $\tau(k)=\frac{1}{4}(k^2-k)$ for all $1\leq k < m$. For any split $m=m_1+m_2$ with $m_1,m_2\geq 1$, we have
% 
% Take any split $m=m_1+m_2$ with $m_1 \geq 1$ and $m_2\geq 1$. Since $m_1<m$ and $m_2<m$, the induction hypothesis applies to both $\tau(m_1)$ and $\tau(m_2)$. Hence,
% \[\tau(m_1)=\frac{1}{4}(m_1^2-m_1),
% \qquad
% \tau(m_2)=\frac{1}{4}(m_2^2-m_2).\]
% 
% Substituting into the recurrence gives
\begin{equation*}
    \begin{aligned}
        \tau(m)&= \tau(m_1)+ \tau(m_2)+ \frac{1}{2}m_1m_2\\
        &= \frac{1}{4}(m_1^2-m_1) + \frac{1}{4}(m_2^2-m_2)+ \frac{1}{2}m_1m_2\\
        &= \frac{1}{4}(m^2-m).
    \end{aligned}
\end{equation*}

Hence, by strong induction, $\tau(m)=\frac{1}{4}(m^2-m)$ for all $m \geq 1$. Consequently, the total cost in Algorithm~\ref{Algo:general-basis} associated with the Taylor expansion is $\frac{1}{4}2^m(m^2-m)$ additions and multiplications.

\noindent We now determine the cost arising from the sub-FFTs. Since the numbers of additions and multiplications coincide, let $S(m)$ denote this common quantity. The algorithm performs $2^{m_1}$ column FFTs of size $2^{m_2}$ and $2^{m_2}$ row FFTs of size $2^{m_1}$. Therefore,
\[S(m) = 2^{m_1} S(m_2) + 2^{m_2} S(m_1), \]
with $S(1)=2$. By the same strong induction argument, this recurrence leads $S(m)=2^m m$ for every split.
% \[S(m)=2^m m. \]

\noindent Combining both contributions, the total number of additions and multiplications performed by Algorithm~\ref{Algo:general-basis} is given by
\[
\frac{1}{4}2^m(m^2-m)+ 2^m m= \frac{1}{4}2^m m^2+ \frac{3}{4}2^m m.
\]

This proves the stated operation count.
\end{proof}

\begin{corollary}\label{Coro:cost-Algo1-subspace}
In the setting of Theorem~\ref{Th:costAlgo1}, when the evaluation domain is the subspace $W_m$, Algorithm~\ref{Algo:general-basis} requires $n-1$ fewer multiplications and $n-1$ fewer additions than over a general affine subspace $\theta+W_m$.
\end{corollary}

\begin{proof}
When $\theta=0$, all $2^{m_1}$ column AFFTs are evaluated over $W_{m_2}$, and exactly one row AFFT is evaluated over $W_{m_1}$.
Thus, if $D(m)$ denotes the number of saved operations, then it satisfies the recurrence
\[
D(m)=2^{m_1}D(m_2)+D(m_1),
\]
with $D(1)=1$. Therefore, by strong induction, $D(m)=2^m-1$. Hence, evaluation over the subspace $W_m$ saves exactly $n-1$ multiplications and $n-1$ additions.
\end{proof}

\begin{remark}[Inverse AFFT]\label{Remark:inverse-AFFT}
The inverse of Algorithm~\ref{Algo:general-basis} is obtained by inverting each stage in reverse order. Specifically, one first applies the inverse row sub-AFFTs over the cosets $\theta_k + W_{m_1}$, then the inverse column sub-AFFTs over the projected affine space $Z_{W_{m_1}}(\theta) + Z_{W_{m_1}}(U)$, and finally the inverse Taylor expansion, which recombines the entries of $M$ into the coefficients of $f(x)$. The resulting polynomial is the unique polynomial of degree less than $2^m$ that interpolates the given evaluations over $\theta+W_m$. The Taylor stage is bijective by the uniqueness in Definition~\ref{Def:TaylorExp}, and the inverse sub-AFFTs are themselves inverse AFFTs of dimensions $m_1$ and $m_2$. Consequently, the inverse transform requires the same number of finite field additions and multiplications as the forward transform. The same holds for Algorithms~\ref{Algo:AFFT-Cantor-any-m1} and~\ref{Algo:AFFT-Cantor-fixed-m1}.
\end{remark}

\noindent The cost of Algorithm~\ref{Algo:general-basis} is dominated by the Taylor expansion step, which requires $O(n(\log_2 n)^2)$ multiplications. This cost arises from the divisions by vanishing polynomials whose coefficients introduce field multiplications. It is therefore natural to ask whether a suitable choice of basis can remove this quadratic overhead. The answer is affirmative. If the affine space is constructed from a Cantor special basis (see Definition~\ref{Def:Cantor-basis}). Then the vanishing polynomial reduces to $S^{m_1}(x)$, and, when $m_1$ is a power of two, it takes the simple form $S^{m_1}(x)=x^{2^{m_1}}+x$. Division by this binomial is multiplication-free and requires only one addition at each iteration. Exploiting this additional structure, the next section presents a specialization of Algorithm~\ref{Algo:general-basis} whose multiplication count is $\frac{1}{2}n\log_2 n$.

\section{New Additive FFTs Based on Cantor Special Basis}\label{Sec:AFFT-Cantor}
We now specialize Algorithm~\ref{Algo:general-basis} to the setting where the evaluation subspace admits a Cantor special basis. In the general basis setting of Section~\ref{Sec:Algo-general}, the Taylor expansion step dominates the overall cost. By Proposition~\ref{Prop:costTE}, computing the Taylor expansion with respect to $Z_{W_{m_1}}(x)$ requires $m_1m_2 2^{m-1}$ multiplications. In contrast, for a Cantor special basis, we have $Z_{W_{m_1}}(x)=S^{m_1}(x)$, and all coefficients of $S^{m_1}(x)$ lie in $\mathbb{F}_2$. Consequently, the Taylor expansion step incurs no multiplication in $\mathbb{F}_{2^k}$. We formalize this observation below.

\begin{remark}\label{Remark:cantorVanishing}
    Let $\{\beta_0=1,\beta_1,\ldots,\beta_{m-1}\}$ be a Cantor special basis of a $m$ dimensional subspace $W_{m}$ of $\mathbb{F}_{2^k}$, and let $m=m_1+m_2$.
    \begin{enumerate}[label=(\roman*)]
        \item Since $S^{m_1}(x)\in\mathbb{F}_2[x]$, the Taylor expansion of any polynomial $f(x)\in\mathbb{F}_{2^k}[x]$ with respect to $S^{m_1}(x)$ requires no multiplications in $\mathbb{F}_{2^k}$.
        
        \item By the explicit formula $S^{m_1}(x)=\sum_{i=0}^{m_1}\binom{m_1}{i}x^{2^i}$ and Lucas' theorem, $S^{m_1}(x)$ has exactly $2^{\mathrm{wt}(m_1)}$ nonzero terms~\cite{BSG2026}, where $\mathrm{wt}(m_1)$ denotes the Hamming weight of $m_1$. Each polynomial-division step in the Taylor expansion costs $2^{\mathrm{wt}(m_1)}-1$ additions. This quantity is minimized when $\mathrm{wt}(m_1)=1$, i.e., $m_1=2^t$ for some positive integer $t$, giving $S^{m_1}(x)=x^{2^{m_1}}+x$ and exactly one addition per step.
    \end{enumerate}
\end{remark}

\noindent Both algorithms below achieve zero multiplications in the Taylor expansion. They differ in the addition cost, which is governed by the number of nonzero terms in $S^{m_1}(x)$. The first algorithm permits an arbitrary split; the second restricts $m_1$ to a power of two, attaining the minimum addition cost by Remark~\ref{Remark:cantorVanishing}(ii).

\noindent Let $W_m=\langle \beta_{0},\beta_{1},\ldots,\beta_{m-1} \rangle$, where $\set{\beta_{0},\beta_{1},\ldots,\beta_{m-1}}$ is a special Cantor basis. We want to evaluate the polynomial $f(x)$ of degree less than $n=2^m$ over the affine space $\theta+W_m$. For a given split $m=m_1+m_2$, we can decompose the subspace $W_m$ into $2^{m_2}$ disjoint affine subspaces $\theta_i+W_{m_1}$, where
\[
\theta_i = \sum_{j=0}^{m_2-1} i_j \beta_{m_1+j}, \quad i=\sum_{j=0}^{m_2-1}i_j 2^j, \quad i_j \in \{0,1\},
\]
and $\theta_0=0$. Also, the recursion gives $S^{m_1}(\beta_{m_1+j})=\beta_{j}$ for $0\le j \le {m_2}-1$, which implies $S^{m_1}(\theta_i)\in W_{m_2}$. Thus, $S^{m_1}$ maps each affine subspace to a single point in $W_{m_2}$.

Taylor expansion of the input $f(x)$ with respect to the vanishing polynomial $S^{m_1}(x)$ gives
\begin{equation*}
\begin{aligned}
f(x)
&=
\sum_{i=0}^{2^{m_2}-1}
\left(
\sum_{j=0}^{2^{m_1}-1} g_{i,j}x^j
\right) (S^{m_1}(x))^i        \\
&=
\sum_{j=0}^{2^{m_1}-1}
\left(
\sum_{i=0}^{2^{m_2}-1} g_{i,j}S^{m_1}(x)^i
\right)x^j
\end{aligned}
\end{equation*}
and we store the coefficients $g_{i,j}$ in a $2^{m_2}\times 2^{m_1}$ matrix $M$. By Remark~\ref{Remark:cantorVanishing}, this step is multiplication-free.

\begin{algorithm2e}[ht]
    \caption{Additive FFT of length $n = 2^m$ over a Cantor special basis with arbitrary split of $m$}\label{Algo:AFFT-Cantor-any-m1}
    \SetAlgoLined
    {\scriptsize
    \KwIn{$f(x) \in \mathbb{F}_{2^k}[x]$ of degree $< n = 2^m$ and the affine subspace $\theta+W_m=\theta+ \langle \beta_0,\beta_1,\ldots,\beta_{m-1} \rangle$, where $\{\beta_0=1,\beta_1,\ldots,\beta_{m-1}\}$ is a Cantor special basis.}
    \KwOut{$\fft(f, \theta+W_m)$.}

    \If{$m=1$}
    {
        \Return $(f(\theta),f(\theta+1))$\;
    }
    Let $m = m_1 + m_2$ and $U=\left\langle \beta_{m_1},\beta_{m_1+1},\ldots,\beta_{m_1+m_2-1} \right\rangle$\;
    
    Do the Taylor Expansion of $f(x)$  w.r.t. $S^{m_1}(x)$ 
    \begin{equation*}
        \begin{aligned}
            f(x)
            &=\sum_{i=0}^{2^{m_2}-1} \left( \sum_{j=0}^{2^{m_1}-1} g_{i,j} x^j \right) (S^{m_1}(x))^i
        \end{aligned}
    \end{equation*}

    Initialize a $2^{m_2} \times 2^{m_1}$ matrix $M$ such that $M[k_2, k_1] = g_{k_2, k_1}$\;

    Let $U' = S^{m_1}(\theta) + W_{m_2}=S^{m_1}(\theta)+  \langle \beta_0, \ldots, \beta_{m_2-1} \rangle $\;
    
    \tcc{Compute column evaluations over the affine space $U'$}
    \For{$k_1 = 0, \dots, 2^{m_1}-1$}{
        Define $C_{k_1}(y) = \sum_{k_2=0}^{2^{m_2}-1} M[k_2, k_1] y^{k_2}$\;
        
        % Let the column polynomial be $C_{k_1}(x) = \sum_{k_2=0}^{2^{m_2}-1} M[k_2, k_1] x^{k_2}$\;
        $M[*, k_1] \leftarrow \fft(C_{k_1}(y), U')$\;
    }

    \tcc{Compute row polynomials over the affine spaces}
    \For{$k_2 = 0, \dots, 2^{m_2}-1$}{
        % Let $(c_{m_2-1}, \ldots, c_0)_2$ be the binary representation of $k_2$\;
        Let $\theta_{k_2} = \theta + \sum_{\ell=0}^{m_2-1} c_{\ell} \beta_{m_1+\ell}$, where $(c_{m_2-1}, \ldots, c_0)_2=k_2$\;
        
        Define $R_{k_2}(x) = \sum_{k_1=0}^{2^{m_1}-1} M[k_2, k_1] x^{k_1}$\;
        
        $E_{k_2} \leftarrow \fft(R_{k_2}(x), \theta_{k_2} + W_{m_1})$\;
    }

    \Return $E_0 \,||\, E_1 \,||\, \dots \,||\, E_{2^{m_2}-1}$
    }
\end{algorithm2e}

\begin{proposition}\label{Prop:cost-Algo1-mult}
Algorithm~\ref{Algo:AFFT-Cantor-any-m1} evaluates a polynomial of degree less than $2^m$ over an $m$-dimensional affine subspace using $\frac{1}{2}n \log_2 n$ multiplications, regardless of the choice of split $m=m_1+m_2$.
\end{proposition}

\begin{proof}
Over a Cantor special basis, the Taylor expansion stage requires no finite field multiplications. Hence, all multiplications arise from the recursive sub-AFFTs. Their number satisfies
\[
M(m)=2^{m_1}M(m_2)+2^{m_2}M(m_1), \qquad M(1)=1.
\]
Using strong induction,
\[
\begin{aligned}
M(m)
&=2^{m_1}\cdot m_2 2^{m_2-1} +2^{m_2} \cdot m_1 2^{m_1-1}\\
&=(m_1+m_2)2^{m-1} =\frac{1}{2}n\log_2 n.
\end{aligned}
\]
Therefore, the multiplication count is independent of the decomposition $m=m_1+m_2$.
\end{proof}

\begin{remark}\label{Remark:cost-Algo1-addi}
    Although the Taylor expansion in Algorithm~\ref{Algo:AFFT-Cantor-any-m1} requires no multiplications for any decomposition $m=m_1+m_2$, its addition count depends on the choice of $m_1$. More precisely, each division step requires $2^{\operatorname{wt}(m_1)}-1$ additions, where $\operatorname{wt}(m_1)$ denotes the Hamming weight of $m_1$. Thus, the addition complexity is determined by the number of nonzero terms in the subspace vanishing polynomial $S^{m_1}(x)$. In particular, if $m_1=m-1$ and $m_2=1$ are chosen recursively, Algorithm~\ref{Algo:AFFT-Cantor-any-m1} yields the same arithmetic cost as the Cantor AFFT~\cite{Cantor1989FFT} reported in Table~\ref{Table:AFFT_cost-cmp}.
\end{remark}

\noindent The preceding remark motivates choosing $m_1$ to have the smallest possible Hamming weight. In particular, choosing $m_1=2^t$ yields $\mathrm{wt}(m_1)=1$, so that $S^{m_1}(x)=x^{2^{m_1}}+x$ and each division step requires only one addition. To preserve this binomial structure throughout the recursion, we choose $m_1=2^{\floor{\log_2(m-1)}}$, namely, the largest power of two strictly smaller than $m$. For instance, when $m=15$, this choice gives $m_1=8$. Algorithm~\ref{Algo:AFFT-Cantor-fixed-m1} formalizes the resulting recursive power-of-two splitting strategy.

% If we do the Taylor expansion of the input $f(x)$ with respect to the vanishing polynomial $S^{m_1}(x)$, then we can write $f(x)$ as
% \[f(x) = \sum_{i=0}^{n_2-1} \left( \sum_{j=0}^{n_1-1} g_{i,j} x^j \right) (S^{m_1}(x))^i, \]
% where $n_1=2^{m_1}$ and $n_2=2^{m_2}$. We can rewrite $f(x)$ as follows
% \[f(x) = \sum_{j=0}^{n_1-1} \left( \sum_{i=0}^{n_2-1} g_{i,j} (S^{m_1}(x))^i \right) x^j. \]

\begin{algorithm2e}[htb]
    \caption{Additive FFT of length $n = 2^m$ over a Cantor special basis with a power-of-two split of $m$}\label{Algo:AFFT-Cantor-fixed-m1}
    \SetAlgoLined
    {\scriptsize
    \KwIn{$f(x) \in \mathbb{F}_{2^k}[x]$ of degree $< n = 2^m$ and the affine subspace $\theta+W_m=\theta+ \langle \beta_0,\beta_1,\ldots,\beta_{m-1} \rangle$, where $\{\beta_0=1,\beta_1,\ldots,\beta_{m-1}\}$ is a Cantor special basis.}
    \KwOut{$\fft(f, \theta+W_m)$.}

    \If{$m=1$}
    {
        \Return $(f(\theta),f(\theta+1))$\;
    }
    Set $m_1$ as the largest power of two strictly less than $m$, and let $m_2 = m - m_1$\;
    
    Do the Taylor Expansion of $f(x)$  w.r.t. $S^{m_1}(x)$ 
    \begin{equation*}
        \begin{aligned}
            f(x)
            &=\sum_{i=0}^{2^{m_2}-1} \left( \sum_{j=0}^{2^{m_1}-1} g_{i,j} x^j \right) (S^{m_1}(x))^i
        \end{aligned}
    \end{equation*}

    Initialize a $2^{m_2} \times 2^{m_1}$ matrix $M$ such that $M[k_2, k_1] = g_{k_2, k_1}$\;

    Let $U' = S^{m_1}(\theta) + W_{m_2}=S^{m_1}(\theta)+  \langle \beta_0, \ldots, \beta_{m_2-1} \rangle $\;
    
    \tcc{Compute column evaluations over the affine space $U'$}
    \For{$k_1 = 0, \dots, 2^{m_1}-1$}{
        Define $C_{k_1}(y) = \sum_{k_2=0}^{2^{m_2}-1} M[k_2, k_1] y^{k_2}$\;
        
        % Let the column polynomial be $C_{k_1}(x) = \sum_{k_2=0}^{2^{m_2}-1} M[k_2, k_1] x^{k_2}$\;
        $M[*, k_1] \leftarrow \fft(C_{k_1}(y), U')$\;
    }

    \tcc{Compute row polynomials over the affine spaces}
    \For{$k_2 = 0, \dots, 2^{m_2}-1$}{
        % Let $(c_{m_2-1}, \ldots, c_0)_2$ be the binary representation of $k_2$\;
        Let $\theta_{k_2} = \theta + \sum_{\ell=0}^{m_2-1} c_{\ell} \beta_{m_1+\ell}$, where $(c_{m_2-1}, \ldots, c_0)_2=k_2$\;
        
        Define $R_{k_2}(x) = \sum_{k_1=0}^{2^{m_1}-1} M[k_2, k_1] x^{k_1}$\;
        
        $E_{k_2} \leftarrow \fft(R_{k_2}(x), \theta_{k_2} + W_{m_1})$\;
    }

    \Return $E_0 \,||\, E_1 \,||\, \dots \,||\, E_{2^{m_2}-1}$
    }
\end{algorithm2e}

\subsection{Detailed Cost Analysis} 
We now establish the arithmetic complexity of Algorithm~\ref{Algo:AFFT-Cantor-fixed-m1}.  Since the Taylor expansion with respect to $S^{m_1}(x)=x^{2^{m_1}}+x$ is multiplication-free, all multiplications arise from the recursive sub-FFTs. The addition cost consists of the contributions from the recursive sub-FFTs and the Taylor expansion steps.

\begin{theorem}\label{Theorem:cost-Algo-fixed-m1}
Let $n=2^m$ and write $m= 2^{p_1}+2^{p_2}+\cdots + 2^{p_w}$ with $p_1> p_2> \cdots > p_w \geq 0$ (so $w=\mathrm{wt}(m)$). Algorithm~\ref{Algo:AFFT-Cantor-fixed-m1} evaluates a polynomial of degree less than $n$ over an $m$-dimensional affine subspace using
\[
\frac{1}{2}\,n\log_2 n
\]
multiplications, and
\begin{equation*}\label{eq:totalAddGenGM}
n\log_2 n \;+\; n\!\left[ \sum_{i=1}^{w} p_i\, 2^{p_i-2} \;+\; \frac{1}{2}\sum_{i=1}^{w} (m \bmod 2^{p_i}) \right]
\end{equation*}
additions. In particular, when $m=2^t$ (i.e.\ $w=1$), the addition count simplifies to
\[
n\log_2 n+\frac{1}{4}n\log_2 n\,\log_2\!\log_2 n.
\]
\end{theorem}

\begin{proof}
The Taylor expansion $f(x)$ with respect to the vanishing polynomial $S^{m_1}(x)$ requires zero finite field multiplications. Thus, the total number of multiplications, denoted by $M(m)$ for an input of size $n = 2^m$, is determined entirely by the sub-FFTs. The algorithm performs $2^{m_1}$ column FFTs of size $2^{m_2}$ and $2^{m_2}$ row FFTs of size $2^{m_1}$. Therefore,
\[M(m) = 2^{m_1} M(m_2) + 2^{m_2} M(m_1), \]
with $M(1)=1$. By the same strong induction argument used in the proof of Theorem~\ref{Th:costAlgo1}, this recurrence yields
\[
M(m)= \frac{1}{2}n \log_2 n.
\]

\noindent Let $A(m)$ denote the total number of additions performed by the algorithm. This cost comprises the additions incurred by the recursive sub-FFTs and those required by the Taylor expansion step, denoted by $A_{\mathrm{TE}}$.

We first analyze the addition cost of the Taylor expansion step, denoted by $A_{TE}$. When $n/d$ is a power of two, the total number of additions required to compute the Taylor expansion of $f(x)$, with degree less than $n$, with respect to $x^d + x$ is $\frac{1}{2}n\log_2 (n/d)$~\cite{Gao2010FFT}.

Since $m_1$ is chosen to be a power of two, we have $S^{m_1}(x)=x^{2^{m_1}}+x$. Applying the above result with $n = 2^{m}$ and $d = 2^{m_1}$, the Taylor expansion of a polynomial of degree less than $2^{m}$ with respect to $S_{m_1}(x)$ requires
\[\frac{1}{2} 2^m \log_2(2^m / 2^{m_1})= \frac{1}{2} 2^m m_2\]
additions. The recursive calls consist of $2^{m_1}$ column FFTs of size $2^{m_2}$ and $2^{m_2}$ row FFTs of size $2^{m_1}$. Therefore, the Taylor expansion cost satisfies
\[A_{TE}(m) = 2^{m_1} A_{TE}(m_2) + 2^{m_2} A_{TE}(m_1) + \frac{1}{2} 2^m m_2. \]

Define $\tau(r)=A_{TE}(r)/2^r$ and dividing the recurrence by $2^m$ yields
\[\tau(m)= \tau(m_1)+ \tau(m_2)+ \frac{1}{2}m_2. \]

We first consider the case in which $m$ is a power of two, say $m=2^t$. The decomposition is then balanced, with $m_1=m_2=2^{t-1}$. Hence, the recurrence for $\tau$ becomes
\[\tau(2^t) = 2 \tau(2^{t-1}) + \frac{1}{2} 2^{t-1}.  \]
Together with the base case $\tau(1)=0$, corresponding to $t=0$, this recurrence yields
\[
\tau(2^t)=t2^{t-2}.
\]

We now derive a closed form for $\tau(m)$ when $m$ is not power of two. Write the binary expansion of $m$ as
\[m= 2^{p_1}+2^{p_2}+\cdots + 2^{p_w}, \]
where $p_1> p_2> \cdots > p_w \geq 0$. Then
\[
m_1=2^{p_1} \quad \text{and} \quad m_2=m-2^{p_1}.
\]

Repeatedly applying the recurrence for $\tau$ gives
\begin{equation*}
    \begin{aligned}
        \tau(m) 
        &= \tau(2^{p_1}) + \tau(m - 2^{p_1}) + \frac{1}{2}(m - 2^{p_1})\\
        &= \tau(2^{p_1})+ \tau \left( \sum_{i=2}^{w}p_i \right)+ \frac{1}{2}(m - 2^{p_1})\\
        &= \tau(2^{p_1})+ \tau(2^{p_2})+ \tau \left( \sum_{i=3}^{w}p_i \right)+ \frac{1}{2}(m - 2^{p_2})+ \frac{1}{2}(m - 2^{p_1})\\
        &= \sum_{i=1}^w \tau(2^{p_i})+ \frac{1}{2} \sum_{i=1}^w (m \bmod 2^{p_i})\\
        &= \sum_{i=1}^w \left( p_i 2^{p_i-2} \right) + \frac{1}{2} \sum_{i=1}^w (m \bmod 2^{p_i}).
    \end{aligned}
\end{equation*}
Consequently, the addition cost for the Taylor expansion is
\[A_{TE}(m)= 2^m \left[ \sum_{i=1}^w \left( p_i 2^{p_i-2} \right) + \frac{1}{2} \sum_{i=1}^w (m \bmod 2^{p_i}) \right]. \]

We now determine the number of additions arising from the sub-FFTs. The algorithm performs $2^{m_1}$ column FFTs of size $2^{m_2}$ and $2^{m_2}$ row FFTs of size $2^{m_1}$. Hence,
\[A_1(m) = 2^{m_1} A_1(m_2) + 2^{m_2} A_1(m_1), \]
with $A_1(1)=2$. By the same strong induction argument used for the multiplication count, this recurrence yields
\[A_1(m)=2^m m. \]

Combining the additions from the recursive sub-FFTs with those from the Taylor expansion steps, the total number of additions is
\begin{equation*}
    \begin{aligned}
        A(m) 
        &=2^m m + 2^m \left[ \sum_{i=1}^w \left( p_i 2^{p_i-2} \right) + \frac{1}{2} \sum_{i=1}^w (m \bmod 2^{p_i}) \right]\\
        &= n\log_2 n \;+\; n\!\left[ \sum_{i=1}^{w} p_i\, 2^{p_i-2} \;+\; \frac{1}{2}\sum_{i=1}^{w} (m \bmod 2^{p_i}) \right]
    \end{aligned}
\end{equation*}

When $m=2^t$, we have $w=1$, $p_1=t$, and $m\bmod 2^{p_1}=0$. It follows that
\[
\tau(m)=t2^{t-2}=\frac{1}{4}m\log_2 m.
\]
Therefore,
\begin{equation*}
    \begin{aligned}
        A(m)
        &= 2^m m + 2^m \cdot \frac{1}{4}m\log_2 m\\
        &= n\log_2 n+\frac{1}{4}n\log_2 n\,\log_2\!\log_2 n.
    \end{aligned}
\end{equation*}

This proves the stated operation count.
\end{proof}

\begin{corollary}\label{Coro:cost-Algo4-subspace}
In the setting of Theorem~\ref{Theorem:cost-Algo-fixed-m1}, when the evaluation domain is the subspace $W_m$, Algorithm~\ref{Algo:AFFT-Cantor-fixed-m1} requires $n-1$ fewer multiplications and $n-1$ fewer additions than over a general affine subspace $\theta+W_m$.
\end{corollary}

\begin{proof}
    When $\theta=0$, all $2^{m_1}$ column AFFTs are evaluated over $W_{m_2}$, and exactly one row AFFT is evaluated over $W_{m_1}$. Thus, if $D(m)$ denotes the number of saved operations, then it satisfies the recurrence
    \[
    D(m)=2^{m_1}D(m_2)+D(m_1),
    \]
    with $D(1)=1$. Therefore, using strong induction, we obtain $D(m)=2^m-1$. Therefore, evaluation over the subspace $W_m$ saves exactly $n-1$ multiplications and $n-1$ additions. 
\end{proof}

\begin{remark}\label{Remark:distinction-GM}
    The splitting rule of Algorithm~\ref{Algo:AFFT-Cantor-fixed-m1} coincides with the balanced decomposition used in the second Gao--Mateer algorithm~\cite{Gao2010FFT} when $m$ is a power of two. Indeed, if $m=2^t$, then 
    \[
    m_1=2^{\floor{\log_2(m-1)}}= 2^{t-1}=\frac{m}{2} \quad \text{and} \quad m_2=m-m_1=\frac{m}{2}.
    \]
    For arbitrary $m$, it generalizes this splitting rule by choosing $m_1$ as the largest power of two strictly smaller than $m$ and setting $m_2=m-m_1$. The resulting decomposition is not necessarily balanced, while the total multiplication count remains $\frac{1}{2}n\log_2 n$. When $m$ is a power of two, the second Gao--Mateer algorithm requires $ n\log_2 n+\frac{1}{2}n\log_2 n\log_2\log_2 n$ additions~\cite{Gao2010FFT}, whereas Algorithm~\ref{Algo:AFFT-Cantor-fixed-m1} requires $n\log_2 n+\frac{1}{4}n\log_2 n\log_2\log_2 n$. Thus, Algorithm~\ref{Algo:AFFT-Cantor-fixed-m1} requires fewer additions than the second Gao--Mateer algorithm.
\end{remark}

\begin{remark}\label{Remark:Cantor-Algo4}
    Cantor's AFFT \cite{Cantor1989FFT} and Algorithm~\ref{Algo:AFFT-Cantor-fixed-m1} both operate over a Cantor special basis and require exactly $\frac{1}{2}n\log_2 n$ multiplications. They differ only in the number of additions, and this difference is attributable entirely to the splitting rule. By Remark~\ref{Remark:cantorVanishing}, each iteration of the division costs $2^{\mathrm{wt}(m_1)}-1$ additions, so the total is governed by the Hamming weights of the successive splits. From Table~\ref{Table:AFFT_cost-cmp}, we know that the addition cost for Cantor algorithm is $\frac{1}{2}n\log_2 n + \frac{1}{2}n\sum_{r=0}^{\log_2 n-1} 2^{\mathrm{wt}(r)}$. When $m=\log_2 n=2^t$, this simplifies to $\tfrac{1}{2}n\log_2 n + \tfrac{1}{2}n\,3^{\log_2\log_2 n}= \tfrac{1}{2}n\log_2 n + \tfrac{1}{2}n (\log_2 n)^{\log_2 3}$. In contrast, by Theorem~\ref{Theorem:cost-Algo-fixed-m1}, Algorithm~\ref{Algo:AFFT-Cantor-fixed-m1} requires $n\log_2 n+\frac{1}{4}n\log_2 n\,\log_2\!\log_2 n$ additions when $m=2^t$. Consequently, Algorithm~\ref{Algo:AFFT-Cantor-fixed-m1} has lower asymptotic addition complexity. Although the above closed forms assume $m=2^t$, the same advantage holds for every $m\geq 4$, since Algorithm~\ref{Algo:AFFT-Cantor-fixed-m1} always chooses $m_1$ with $\mathrm{wt}(m_1)=1$, thereby minimizing the addition cost at each step.
\end{remark}

\begin{remark}\label{Remark:distinction-LCH}
    Over a Cantor special basis, both Algorithm~\ref{Algo:AFFT-Cantor-fixed-m1} and the LCH AFFT~\cite{LCH-FFT2016} require $\frac{1}{2}n\log_2 n$ multiplications. Nevertheless, the two algorithms differ substantially in their treatment of the input polynomial. Algorithm~\ref{Algo:AFFT-Cantor-fixed-m1} operates directly on the polynomial in the standard monomial basis. In contrast, the LCH algorithm first converts the input from the monomial basis to the novel polynomial basis and then performs the evaluation using its butterfly procedure. Consequently, Algorithm~\ref{Algo:AFFT-Cantor-fixed-m1} avoids the memory-access overhead associated with the basis-conversion stage of the LCH algorithm. The performance improvement observed over the LCH algorithm is reported in Table~\ref{Table:benchmark-results}.
\end{remark}

% \begin{remark}\label{Remark:cost-Algo4-subspace}
% The complexity in Theorem~\ref{Theorem:cost-Algo-fixed-m1} applies to evaluation over a affine space $\theta+W_m$. When working with subspace, i.e., when $\theta=0$, some additions and multiplications can be avoided. Indeed, the evaluation space for every column AFFT becomes $W_{m_2}$. one row AFFT is evaluated over the subspace $W_{m_1}$. The other $2^{m_2}-1$ row AFFTs are evaluated over affine spaces with $\theta_{k_2}\neq 0$. Let $D(m)$ be the number of multiplications, and also the number of additions, saved in dimension $m$. Then it satisfies the recurrence
% \[
% D(m)=2^{m_1}D(m_2)+D(m_1),
% \]
% with $D(1)=1$. Therefore, using strong induction, we obtain $D(m)=2^m-1$. Therefore, evaluation over the subspace $W_m$ saves exactly $n-1$ multiplications and $n-1$ additions. 
% \end{remark}

% \noindent Algorithms~\ref{Algo:general-basis} and \ref{Algo:AFFT-Cantor-fixed-m1} both operate directly on polynomials represented in the standard monomial basis and use Taylor expansion to decompose the evaluation problem. The next section considers the corresponding decomposition problem for polynomials represented in the novel polynomial basis introduced by Lin, Chung, and Han~\cite{LCH-basis2014}.

\section{Generalized Butterfly Phase of the LCH Additive FFT}\label{Sec:LCH-butterfly}

The arbitrary matrix decomposition introduced in Algorithm~\ref{Algo:general-basis} suggests a corresponding generalization of the evaluation phase of the LCH AFFT~\cite{LCH-basis2014,LCH-FFT2016}. Specifically, the column-row viewpoint can be transferred to polynomials represented in the novel polynomial basis, allowing the LCH butterfly phase to use any decomposition $m=m_1+m_2$.
This result is separate from Algorithms~\ref{Algo:AFFT-Cantor-any-m1} and \ref{Algo:AFFT-Cantor-fixed-m1}. They operate directly on coefficients in the standard monomial basis, whereas the generalized LCH butterfly developed in this section assumes that the input has already been converted to the novel polynomial basis~\cite{LCH-basis2014}. Thus, Algorithm~\ref{Algo:GenLCH} generalizes only the LCH evaluation phase; it does not eliminate the preceding polynomial-basis conversion.

The LCH additive FFT represents the input polynomial in the novel polynomial basis introduced by Lin, Chung, and Han~\cite{LCH-basis2014}. Although the algorithm is independent of the choice of basis for the underlying finite field, we assume a Cantor special basis throughout this section for simplicity.

\begin{definition}
Let $B=\{1,\beta_1,\ldots,\beta_{m-1}\}$ be a Cantor special basis of a $m$ dimensional subspace $W_{m}$ of $\mathbb{F}_{2^k}$, and let $S^i(x)$ denote the corresponding subspace polynomials. For each integer
\[
k=\sum_{i=0}^{m-1} b_i2^i,
\qquad b_i\in\{0,1\},
\]
define
\[
X_k(x)\coloneqq\prod_{i=0}^{m-1}\bigl(S^i(x)\bigr)^{b_i}.
\]
The family $\{X_0(x),X_1(x),\ldots,X_{n-1}(x)\}$ is called the novel polynomial basis with respect to $B$.
\end{definition}

Thus, $X_k(x)$ is the product of the polynomials $S^i(x)$ for which the $i$th bit of $k$ is equal to one. Since $\deg(S^i(x))=2^i$, it follows that $\deg(X_k(x))=k$ for every $k\in\{0,1,\ldots,n-1\}$. Hence, any polynomial $f(x)$ of degree less than $n$ admits a unique representation
\[
f(x)=\sum_{k=0}^{n-1} f_k X_k(x),
\]
where $f_0,f_1,\ldots,f_{n-1}\in\mathbb{F}_{2^m}$. The LCH additive FFT takes the coefficients of $f(x)$ in this novel polynomial basis as input and evaluates $f$ through a butterfly decomposition.

The proposed butterfly phase relies on the compositional structure of the vanishing polynomials associated with a Cantor special basis. In particular, these polynomials satisfy
\[
S^{i+j}(x)=S^i\bigl(S^j(x)\bigr).
\]
This identity allows us to factor the novel polynomial basis $X_k(x)$ to be decomposed into factors of smaller degree. 

Let $P(x)$ a polynomial of degree less than $n=2^m$, where $m=m_1+m_2$, in novel polynomial basis
\[
P(x) = \sum_{k=0}^{2^m-1} p_k X_k(x).
\]
Each index $k\in\{0,\ldots,2^m-1\}$ can be written uniquely as
\[k = k_1 + 2^{m_1}\cdot k_2, \]
where $k_1$ ranges from $0$ to $2^{m_1}-1$ and $k_2$ ranges from $0$ to $2^{m_2}-1$. Accordingly,
\[P(x) = \sum_{k_1=0}^{2^{m_1}-1} \sum_{k_2=0}^{2^{m_2}-1} p_{k_1 + 2^{m_1}\cdot k_2} X_{k_1 + 2^{m_1}\cdot k_2}(x).\]

The binary supports of $k_1$ and $2^{m_1}k_2$ are disjoint, and hence
\begin{equation*}
    \begin{aligned}
       X_{k_1 + 2^{m_1}\cdot k_2}(x)&= X_{k_1}(x)\cdot X_{2^{m_1}\cdot k_2}(x).
       % \\
       % &= X_{k_1}(x) \cdot X_{k_2}(s^{m_1}(x))
    \end{aligned}
\end{equation*}

Write $k_2 = \sum_{j=0}^{m_2-1} b'_j 2^j$. Then
\begin{equation*}
    \begin{aligned}
        2^{m_1} \cdot k_2 
        &= 2^{m_1} \sum_{j=0}^{m_2-1} b'_j 2^j = \sum_{j=0}^{m_2-1} b'_j 2^{j+m_1},
    \end{aligned}
\end{equation*}
which gives
% Thus, we have
% % \[X_{2^{m_1} \cdot k_2}(x) = \prod_{i=m_1}^{m_1+m_2-1} (S^i(x))^{b'_{i-m_1}}  \]
% 
\begin{equation*}
    \begin{aligned}
        X_{2^{m_1} \cdot k_2}(x) 
        % &= \prod_{i=m_1}^{m_1+m_2-1} (S^i(x))^{b'_{i-m_1}}\\
        &= \prod_{j=0}^{m_2-1} (S^{j+m_1}(x))^{b'_j}\\
        &= \prod_{j=0}^{m_2-1} \left( S^j(S^{m_1}(x)) \right)^{b'_j}.
    \end{aligned}
\end{equation*}

Therefore,
\[X_{k_1 + 2^{m_1}\cdot k_2}(x)= X_{k_1}(x)\cdot  X_{k_2}(S^{m_1}(x)).\]

Substituting this identity into the representation of $P(x)$ yields
\begin{equation*}
    \begin{aligned}
        P(x) &= \sum_{k_1=0}^{2^{m_1}-1} \sum_{k_2=0}^{2^{m_2}-1} p_{k_1 + 2^{m_1}\cdot k_2} X_{k_1}(x)\cdot  X_{k_2}(S^{m_1}(x))\\
        &= \sum_{k_1=0}^{2^{m_1}-1} \left[ \sum_{k_2=0}^{2^{m_2}-1} p_{k_1 + 2^{m_1}\cdot k_2} X_{k_2}(S^{m_1}(x)) \right] X_{k_1}(x).
    \end{aligned}
\end{equation*}

\noindent We can decompose the subspace $W_m$ into $2^{m_2}$ disjoint affine subspaces $\theta_i+W_{m_1}$, where $0\le i 
\le 2^{m_2}-1$ and $\theta_i$ is linear combination of $\set{\beta_{m_1},\beta_{m_1+1},\ldots,\beta_{m_1+m_2-1}}$ with $\theta_0=0$. Also, note that $S^{m_1}(\beta_{m_1+j})=\beta_{j}$ for $0\le j \le {m_2}-1$. Therefore, for each $\theta_i+\omega \in \theta_i+W_{m_1}$, we have $S^{m_1}(\theta_i+\omega)=S^{m_1}(\theta_i)$ which implies that for each $x\in \theta_i+W_{m_1}$, we have $S^{m_1}(x) \in W_{m_2}$.

With the factorization of $P(x)$ and the associated affine space structure in place, the generalized butterfly phase follows naturally. The algorithm first evaluates the $2^{m_1}$ column polynomials over the projected space $U$ and then evaluates each of the resulting $2^{m_2}$ row polynomials over its corresponding affine space. Algorithm~\ref{Algo:GenLCH} formalizes this procedure for an arbitrary decomposition $m=m_1+m_2$.

\begin{algorithm2e}[htb]
    \caption{Generalized butterfly phase of LCH additive FFT of length $n = 2^m$}\label{Algo:GenLCH}
    \SetAlgoLined
    {\scriptsize
    \KwIn{Polynomial $P(x) = \sum_{k=0}^{n-1} p_k X_k(x)$ in the novel polynomial basis, and the affine subspace $\theta+W_m=\theta+ \langle \beta_0,\beta_1,\ldots,\beta_{m-1} \rangle$, where $\{\beta_0=1,\beta_1,\ldots,\beta_{m-1}\}$ is a Cantor special basis.}
    \KwOut{$\fft(P, \theta+W_m)$.}
    
    \If{$m=1$}
    {
        \Return $(P(\theta),P(\theta+1))$\;
    }
    
    Write $P(x)$ as
    \[
        P(x) = \sum_{k_1=0}^{2^{m_1}-1} \left[ \sum_{k_2=0}^{2^{m_2}-1} p_{k_1 + 2^{m_1}\cdot k_2} X_{k_2}(S^{m_1}(x)) \right] X_{k_1}(x)
    \]

    Initialize a $2^{m_2} \times 2^{m_1}$ matrix $M$ such that $M[k_2, k_1] = p_{k_1 + 2^{m_1}\cdot k_2}$\;

    Let $U' = S^{m_1}(\theta) + W_{m_2} = S^{m_1}(\theta) + \langle \beta_0, \ldots, \beta_{m_2-1} \rangle$\;
    
    \tcc{Evaluate column polynomials over the affine space $U'$}
    \For{$k_1 = 0, \dots, 2^{m_1}-1$}{
        Define $C_{k_1}(y) = \sum_{k_2=0}^{2^{m_2}-1} M[k_2, k_1] X_{k_2}(y)$\;
        
        $M[*, k_1] \leftarrow \fft(C_{k_1}(y), U')$\;
    }

    \tcc{Evaluate row polynomials over the affine spaces}
    \For{$k_2 = 0, \dots, 2^{m_2}-1$}{
        Let $\theta_{k_2} = \theta + \sum_{\ell=0}^{m_2-1} c_{\ell} \beta_{m_1+\ell}$, where $(c_{m_2-1}, \ldots, c_0)_2=k_2$\;

        Define $R_{k_2}(x) = \sum_{k_1=0}^{2^{m_1}-1} M[k_2, k_1] X_{k_1}(x)$\;
        
        $E_{k_2} \leftarrow \fft(R_{k_2}(x), \theta_{k_2} + W_{m_1})$\;
    }

    \Return $E_0 \,||\, E_1 \,||\, \dots \,||\, E_{2^{m_2}-1}$\;
    }
\end{algorithm2e}

% \subsection{Cost analysis}\label{GenLCH:cost}

\noindent The arithmetic cost of Algorithm~\ref{Algo:GenLCH} is invariant under the choice of split $m=m_1+m_2$, as the following result shows.

\begin{theorem}\label{Th:costGenLCH}
    Algorithm~\ref{Algo:GenLCH} evaluates a polynomial of degree less than $n=2^m$, represented in the novel polynomial basis, over an $m$-dimensional affine subspace using exactly
    \[
    \frac{1}{2}n\log_2 n
    \quad \text{multiplications and}\quad
    n\log_2 n
    \quad\text{additions},
    \]
    independently of the decomposition $m=m_1+m_2$.
\end{theorem}

\begin{proof}
Let $A(m)$ and $M(m)$ denote the total number of additions and multiplications performed by the algorithm, respectively. We prove by strong induction on $m$ that
\[
A(m)=m2^m \qquad \text{and} \qquad M(m)=\frac{m}{2}2^m.
\]

For the base case $m=1$, the algorithm uses $A(1)=2$ additions and $M(1)=1$ multiplication. Hence,
\[
A(1)=1\cdot 2^1 \qquad \text{and} \qquad M(1)=\frac{1}{2}\cdot 1\cdot 2^1.
\]
% For the balanced split $m_1=m_2=m/2$, the cost of the algorithm satisfies the recurrences
% \begin{equation*}
% \begin{aligned}
%     A(m)&= 2^{m/2}A(m/2)+ 2^{m/2}A(m/2)\\
%     M(m)&= 2^{m/2}M(m/2)+ 2^{m/2}M(m/2),
% \end{aligned}
% \end{equation*}
% with base cases $A(1)=2$ and $M(1)=1$. Solving these recurrences yields that the total number of additions is $A(m)=2^m m $, and the total number of multiplications is $M= \frac{1}{2} 2^m m$.

% We now show, by strong induction, that $A(m)=m\cdot 2^m$ for an arbitrary decomposition $m=m_1+m_2$. The argument for $M(m)=\frac{m}{2}\cdot 2^m$ is identical.

Now let $m\geq 2$, and assume that, for every positive integer $k<m$, $A(k)=k2^k$ and $M(k)=\frac{1}{2}k2^k$.

Consider an arbitrary decomposition $m=m_1+m_2$, where $m_1,m_2\geq 1$. The algorithm performs $2^{m_1}$ recursive column FFTs of size $2^{m_2}$ and $2^{m_2}$ recursive row FFTs of size $2^{m_1}$. Therefore,
\[
A(m)=2^{m_1}A(m_2)+2^{m_2}A(m_1).
\]

Applying the induction hypothesis gives
\begin{align*}
A(m)
&=2^{m_1}\bigl(m_2 2^{m_2}\bigr)
  +2^{m_2}\bigl(m_1 2^{m_1}\bigr)\\
&=(m_1+m_2)2^{m_1+m_2} =m2^m.
\end{align*}

The same argument gives $M(m)=\frac{1}{2}m 2^m$. Therefore, the algorithm uses $n\log_2 n$ additions and $\frac{1}{2}n\log_2 n$ multiplications, independently of the decomposition $m=m_1+m_2$.
\end{proof}

% \noindent We note that that Algorithm~\ref{Algo:GenLCH} generalizes the butterfly phase of the LCH AFFT~\cite{LCH-FFT2016}, which is obtained by choosing $m_1=m-1$ and $m_2=1$.

\begin{remark}
Algorithm~\ref{Algo:GenLCH} generalizes the butterfly phase of the LCH AFFT, recovered by choosing $m_1=m-1$ and $m_2=1$. Although this generalization preserves the arithmetic cost of the LCH butterfly phase for an arbitrary decomposition, it does not eliminate the polynomial basis conversion required by the LCH framework. Consequently, Algorithm~\ref{Algo:AFFT-Cantor-fixed-m1} retains its principal implementation advantage by accepting coefficients in the standard monomial basis and avoiding the memory-access overhead associated with the conversion to the novel polynomial basis.
\end{remark}

\noindent The algorithms developed so far, namely Algorithm~\ref{Algo:general-basis} in the general basis setting and Algorithms~\ref{Algo:AFFT-Cantor-fixed-m1} and \ref{Algo:GenLCH} in the full Cantor special basis setting, assume either that the evaluation subspace has no additional structure or that it admits a complete Cantor special basis. However, a binary extension field may admit a partial Cantor special basis. More precisely, a prefix of the ordered basis satisfies the Cantor recursion (see Definition~\ref{Def:Cantor-basis}), whereas the remaining basis vectors are generic. The next section characterizes the parameter regimes in which this advantage arises.

\section{Additive FFTs in Partial Cantor special basis}\label{Sec:partial-Cantor}
The first Gao--Mateer~\cite{Gao2010FFT} and von zur Gathen--Gerhard~\cite{zurGathenFFT} algorithms evaluate a polynomial $f(x)$ over an $m$-dimensional subspace of $\mathbb{F}_{2^k}$, where $k$ is arbitrary. From Proposition~\ref{Prop:partial-cantor}, we know that $\mathbb{F}_{2^k}$ admits a Cantor special basis of dimension $\ell$ if and only if it contains $\mathbb{F}_{2^{2^{\ceil{\log_2 \ell}}}}$ as a subfield. Here, we focus on how a partial Cantor special basis changes the arithmetic costs of the von zur Gathen--Gerhard AFFT and Algorithm~\ref{Algo:general-basis}. For completeness, Appendix~\ref{Sec:GG-Cantor-AFFT} reviews Cantor's AFFT~\cite{Cantor1989FFT} and the von zur Gathen--Gerhard generalization.

If $\mathbb{F}_{2^k}$ does not admit a Cantor special basis of the full evaluation dimension $m$, the first Gao--Mateer additive FFT is generally the preferred method. We show, however, that the von zur Gathen--Gerhard algorithm and Algorithm~\ref{Algo:general-basis} can outperform it when a partial Cantor special basis is available. In particular, both algorithms can exploit this structure to reduce their numbers of additions and multiplications, whereas the Gao--Mateer algorithm cannot obtain the same reduction because its basis changes at each iteration. We characterize the parameter regimes in which these savings compensate for the higher general cost of the two algorithms, with particular emphasis on the case where the dimension of the evaluation subspace exceeds the dimension of the available Cantor special basis.

% We demonstrate that in certain scenarios, the von zur Gathen--Gerhard algorithm and Algorithm~\ref{Algo:general-basis} outperforms the first Gao--Mateer method. Our key contribution lies in identifying the role of partial Cantor special bases and showing that when such a basis is available, both the von zur Gathen--Gerhard algorithm and Algorithm~\ref{Algo:general-basis} can leverage it to reduce the number of additions and multiplications. In contrast, the Gao--Mateer algorithm cannot take advantage of this reduction due to its changing basis at each iteration. We characterize the conditions under which the von zur Gathen--Gerhard algorithm and Algorithm~\ref{Algo:general-basis} outperform the first
% Gao--Mateer method, particularly when the dimension of the evaluation subspace exceeds the dimension of the available Cantor special basis.

\subsection{Comparing Gao--Mateer and von zur Gathen--Gerhard AFFTs}\label{Sec:partial-GMvGG}
In this section, we first analyze the exact computational cost of the von zur Gathen--Gerhard AFFT~\cite{zurGathenFFT}. The vanishing polynomial $\mathbb{Z}_{W{m-1}}(x)$ is linearized, taking the form $\mathbb{Z}_{W{m-1}}(x) = c_0x + c_1x^2 + \cdots +c_{m-1}x^{2^{m-2}}+ x^{2^{m-1}}$, and consequently, $\mathbb{Z}_{W{m-1}}(x+\theta)$ contains $m+1$ coefficients. Since the highest-degree term in $\mathbb{Z}_{W_{m-1}}(x+\theta)$ has a coefficient of one, each iteration of dividing $f(x)$ by $\mathbb{Z}_{W_{m-1}}(x+\theta)$ requires $m$ multiplications. Additionally, at each iteration, the highest-degree term cancels out, allowing us to omit the corresponding addition. Consequently, dividing $f(x)$ by $\mathbb{Z}_{W_{m-1}}(x+\theta)$ requires a total of $2^{m-1} \cdot m$ multiplications and additions. 

Let $f(x) = q(x) \mathbb{Z}_{W_{m-1}}(x+\theta) + f_0(x)$. Then, we can rewrite it as  
\[
f(x) = q(x) \mathbb{Z}_{W_{m-1}}(x + \theta + \beta_{m-1}) + f_0(x) + q(x) \mathbb{Z}_{W_{m-1}}(\beta_{m-1}).
\]
Since the degree of $f(x)$ is less than $n=2^m$ and the degree of $\mathbb{Z}_{W_{m-1}}(x+\theta)$ is $2^{m-1}$, it follows that both $q(x)$ and $f_0(x)$ have degrees less than $2^{m-1}$. This allows us to define $f_1(x) = f_0(x) + q(x) \mathbb{Z}_{W_{m-1}}(\beta_{m-1})$. Thus, to compute $f_1(x)$, we require $2^{m-1}$ multiplications and additions. Thus, obtaining both $f_0(x)$ and $f_1(x)$ requires a total of $2^{m-1} \cdot m+ 2^{m-1}= 2^{m-1} \cdot (m+1)$ multiplications and additions. Similarly, for the next round we need $2 \cdot (2^{m-2}m)=2^{m-1}m$ multiplications and additions. More generally, at the $r$th iteration, we need $2^{m-1}(m+1-r)$ many multiplications and additions. Therefore, the total number of multiplications and additions in von zur Gathen--Gerhard FFT is given by
% 
% \begin{equation}\label{Eqn:zurGathen_M}
%     \begin{aligned}
%         &\sum_{r=0}^{m-1} 2^{m-1}(m+1-r)
%         = 2^{m-1} \big[(m+1)m- \frac{(m-1)m}{2} \big] \\
%         &= \frac{1}{4}2^m m^2+ \frac{3}{4}2^m m= \frac{1}{4} n (\log_2 n)^2 + \frac{3}{4} n \log_2 n.
%     \end{aligned}
% \end{equation}
% 
\begin{align}
&\sum_{r=0}^{m-1} 2^{m-1}(m+1-r)
= 2^{m-1} \big[(m+1)m- \frac{(m-1)m}{2} \big] \notag \\
&= \frac{1}{4}2^m m^2+ \frac{3}{4}2^m m= \frac{1}{4} n (\log_2 n)^2 + \frac{3}{4} n \log_2 n \label{Eqn:zurGathen_M}
\end{align}

% We can observe that the number of additions in both of the algorithms are same. 

Now, we will discuss the case when we have partial Cantor special basis in $\mathbb{F}_{2^k}$. Suppose that $f(x) \in \mathbb{F}_{2^k}[x]$ be a polynomial of degree less than $n=2^m$ and we want to evaluate $f(x)$ over the affine subspace $\theta+W_m=\theta+\langle \beta_0,\beta_1,\ldots, \beta_{\ell-1},\beta_{\ell},\ldots, \beta_{m-1} \rangle$, where $\set{\beta_0,\beta_1,\ldots, \beta_{\ell-1}}$ is a set of Cantor special basis in the field $\mathbb{F}_{2^k}$.

\begin{theorem}\label{Th:cost-mult-GG-partial}
Let $n=2^m$, where $m\geq 3$, , and let $\theta+W_m =\theta+\langle\beta_0,\ldots,\beta_{m-1}\rangle \subseteq\mathbb{F}_{2^k}$ be an $m$-dimensional affine subspace. Suppose that the first $\ell$ basis elements form a Cantor special basis of maximum dimension admitted by $\mathbb{F}_{2^k}$. Then the von zur Gathen--Gerhard AFFT requires
\[
2^{m-2}\left(m^2+3m-\ell^2-3\ell\right).
\]
finite field multiplications to evaluate a polynomial of degree less than $2^m$ over $\theta+W_m$. Moreover, this requires fewer multiplications than the first Gao--Mateer AFFT if and only if $\ell\geq m-2$.
\end{theorem}

\begin{proof}
Let $f(x)\in\mathbb{F}_{2^k}[x]$ have degree less than $n=2^m$, and suppose that $f(x)$ is evaluated over $\theta+W_m$. We divide the multiplication cost of the von zur Gathen--Gerhard AFFT into three parts.

% In this case, after $m-1-(\ell+1)$th iteration, we have $\mathbb{Z}_{W_{i}}(x)=S^{i}(x) \text{ for } i=0,1,\ldots,\ell$. Thus, division by $\mathbb{Z}_{W_{i}}(x) \text{ for } i=0,1,\ldots,\ell$, only requires additions which gives the advantages in the von zur Gathen--Gerhard AFFT. Whereas, Gao--Mateer AFFT can't take that advantage as in each iteration, the evaluation is done over a different subspace, i.e., the basis is changed after each iteration in that algorithm.

% Now, we will compute the multiplication cost of the von zur Gathen--Gerhard algorithm. 

For the iterations indexed by $r=0,\ldots,m-1- (\ell+1)$, the algorithm operates on the general part of the ordered basis. At iteration $r$, it requires $2^{m-1}(m+1-r)$ finite field multiplications. Hence, the total multiplication cost
of these iterations is given by
\begin{equation*}
    \begin{aligned}
        &\sum_{r=0}^{m-1- (\ell+1)} 2^{m-1}(m+1-r)
        % &= 2^{m-1}[(m+1)(m-\ell-1)- \frac{(m-\ell-2)(m-\ell-1)}{2}]\\
        = 2^{m-2}(m-\ell-1)(m+ \ell+4).
    \end{aligned}
\end{equation*}

At the $(m-1-\ell)$th iteration, there are $2^{m-1-\ell}$ polynomials, each of degree less than $2^{\ell+1}$. The divisions in this step are performed by $S^{\ell}(x+\theta+\omega_i)$, for some $\omega_i \in \langle \beta_{\ell},\beta_{\ell+1}, \ldots, \beta_{m-1} \rangle$. Thus, each iteration of the division by $S^{\ell}(x+\theta+\omega_i)$ requires only one multiplication. Furthermore, since $S^{\ell}(\beta_{\ell}) \neq 1$, each of the $2^{m-\ell}$ divisions requires $2^{\ell}$ multiplications. As a result, the total number of multiplications needed for the $(m-1-\ell)$th iteration is
\[
2^{m-\ell}\cdot 2^\ell=2^m.
\]

After the $(m-1-\ell)$-th iteration, the algorithm reduces to the Cantor algorithm. At this stage, there are $2^{m-\ell}$ polynomials, each of degree less than $2^{\ell}$. Consequently, there are $2^{m-\ell}$ Cantor AFFT calls, with each call requiring a multiplication cost of $\frac{1}{2} 2^{\ell} \ell$ (see Table~\ref{Table:AFFT_cost-cmp}). This results in a total of 
\[
2^{m-\ell} \cdot \frac{1}{2} 2^{\ell} \ell = 2^{m-1} \ell
\]
multiplications. Therefore, the overall multiplication cost for the algorithm is given by  

\begin{equation}\label{Eqn:MixedFFT_M}
    \begin{aligned}
        & 2^{m-2}(m-\ell-1)(m + \ell + 4) + 2^{m} + 2^{m-1} \ell \\
        &= 2^{m-2}(m^2 + 3m - \ell^2 - 3\ell).
    \end{aligned}
\end{equation}
This proves the stated multiplication count.

Now, we want to show for which values of $\ell$, for which the multiplication cost in von zur Gathen--Gerhard algorithm will be less than the Gao--Mateer algorithm. For simplicity, we assume that the multiplication cost  in Gao--Mateer multiplication cost is $\frac{3}{2}2^m m - 2^m$. Thus,
\begin{equation*}
    \begin{aligned}
        & 2^{m-2}(m^2+3m-\ell^2-3\ell) \leq \frac{3}{2}2^m m - 2^m\\
        \iff & m^2+3m-\ell^2-3\ell\leq 6m - 4\\
        \iff &\ell^2+ 3\ell - m^2 +3m -4 \geq 0 \\
        \iff & \ell \geq m-2 \quad \text{for } m\geq 3
    \end{aligned}
\end{equation*}

Consequently, whenever $\ell\geq m-2$, the von zur Gathen--Gerhard AFFT requires strictly fewer multiplications than the first Gao--Mateer AFFT.
\end{proof}

\begin{remark}
    If the ordered basis of $W_m$ contains no nontrivial Cantor special prefix, then $\ell = 0$. Substituting $\ell=0$ into (\ref{Eqn:MixedFFT_M}), gives $\frac{1}{4} 2^m m^2 + \frac{3}{4} 2^m m$, which agrees with the multiplication cost in (\ref{Eqn:zurGathen_M}).
\end{remark}

\noindent For the multiplication comparison, the preceding inequality shows that the von zur Gathen--Gerhard AFFT requires fewer multiplications than the first Gao--Mateer AFFT whenever $\ell \geq m-2$, where $\ell$ is the maximum available dimension of the Cantor special basis. Table~\ref{Tab:GM_partial-cantor} reports the exact addition and multiplication counts for several finite fields and for the values of $m$ satisfying this condition.

% Table~\ref{Tab:GM_partial-cantor} compares the addition and multiplication costs of the von zur Gathen--Gerhard and first Gao--Mateer AFFTs for several finite fields. For each field, $\ell$ denotes the maximum available dimension of the Cantor special basis, and the listed values of $m$ satisfy $\ell\geq m-2$. Consequently, the von zur Gathen--Gerhard AFFT requires fewer multiplications than the exact first Gao--Mateer AFFT in every listed case. The table also shows that its addition count is lower for these parameters. Thus, for the cases listed in Table~\ref{Tab:GM_partial-cantor}, the partial Cantor structure reduces both arithmetic costs.

\begin{table}[htb]
    \centering
    \renewcommand{\arraystretch}{1.2}
    \caption{Comparison of the numbers of finite-field additions ($\#A$) and multiplications ($\#M$) required by the first Gao--Mateer and von zur Gathen--Gerhard AFFTs of length $2^m$. Here, $\ell$ is the maximum available dimension of the Cantor special basis over the finite field.}
    \label{Tab:GM_partial-cantor}
    \begin{tabular}{|c|c|c|cc|cc|}
        \hline
        \multirow{2}{*}{Finite field}
        & \multirow{2}{*}{$\ell$}
        & \multirow{2}{*}{$m$}
        & \multicolumn{2}{c|}{First Gao--Mateer}
        & \multicolumn{2}{c|}{von zur Gathen--Gerhard} \\
        \cline{4-7}
        & & & $\#A$ & $\#M$ & $\#A$ & $\#M$ \\
        \hline
        \multirow{2}{*}{$\mathbb{F}_{2^{10}}$}
        & \multirow{2}{*}{2}
        & 3  & 36 & 29 & 32 & 16 \\
        & & 4  & 112 & 81 & 104 & 72 \\
        \hline
        \multirow{2}{*}{$\mathbb{F}_{2^{12}}$}
        & \multirow{2}{*}{4}
        & 5  & 320 & 209 & 256 & 96 \\
        & & 6  & 864 & 513 & 736 & 416 \\
        \hline
        \multirow{2}{*}{$\mathbb{F}_{2^{24}}$}
        & \multirow{2}{*}{8}
        & 9  & 13824 & 6401 & 9728 & 2560 \\
        & & 10 & 33280 & 14337 & 25088 & 10752 \\
        \hline
        \multirow{2}{*}{$\mathbb{F}_{2^{48}}$}
        & \multirow{2}{*}{16}
        & 17 & 11141120 & 3211265 & 6553600 & 1179648 \\
        & & 18 & 24772608 & 6815745 & 15597568 & 4849664 \\
        \hline
    \end{tabular}
\end{table}

The results in Table~\ref{Tab:GM_partial-cantor} show that the von zur Gathen--Gerhard AFFT requires fewer additions and multiplications than the first Gao--Mateer AFFT in every listed case. For example, over $\mathbb{F}_{2^{10}}$ with $\ell=2$ and $m=4$, the multiplication count decreases from $81$ to $72$, while the addition count decreases from $112$ to $104$. The savings become more substantial for larger parameters. Over $\mathbb{F}_{2^{48}}$ with $\ell=16$ and $m=18$, the multiplication count decreases from $6{,}815{,}745$ to $4{,}849{,}664$, and the addition count decreases from $24{,}772{,}608$ to $15{,}597{,}568$. Thus, a partial Cantor special basis can reduce both arithmetic costs when it does not span the entire evaluation subspace.

% \begin{table}[htb]
%     \centering
%     \renewcommand{\arraystretch}{1.2}
%     \caption{Comparison of addition ($A$) and multiplication ($M$) costs in the FFT of length $2^m$ over different finite fields.}
%     \label{Tab:GM_partial-cantor}
%     \renewcommand{\arraystretch}{1.2}
%     \begin{tabular}{|c|c|cc|cc|}
%         \hline
%         {Finite Field} & $m$ & \multicolumn{2}{c|}{{Gao--Mateer}} & \multicolumn{2}{c|}{{von zur Gathen--Gerhard}} \\
%         \cline{3-6}
%         & & \#A & \#M & \#A & \#M \\
%         \hline
%         \multirow{2}{*}{$\mathbb{F}_{2^{10}}$} 
%         & 3  & 36 & 29 & 32 & 16 \\
%         & 4  & 112 & 81 & 104 & 72 \\
%         \hline
%         \multirow{2}{*}{$\mathbb{F}_{2^{12}}$} 
%         & 5  & 320 & 209 & 256 & 96 \\
%         & 6  & 864 & 513 & 736 & 416 \\
%         \hline
%         \multirow{2}{*}{$\mathbb{F}_{2^{24}}$} 
%         & 9  & 13824 & 6401 & 9728 & 2560 \\
%         & 10  & 33280 & 14337 & 25088 & 10752 \\
%         \hline
%         \multirow{2}{*}{$\mathbb{F}_{2^{48}}$}
%         & 17 & 11141120 & 3211265 & 6553600 & 1179648 \\
%         & 18 & 24772608 & 6815745 & 15597568 & 4849664 \\
%         \hline
%     \end{tabular}
% \end{table}

\subsection{Comparing Gao--Mateer and Algorithm~\ref{Algo:general-basis}}\label{Sec:partial-GMvAlgo1}
Following the comparison between the von zur Gathen--Gerhard and first Gao--Mateer algorithms, we now compare our proposed Algorithm~\ref{Algo:general-basis} with the first Gao--Mateer algorithm under the same partial Cantor-basis setting. Recall that, for evaluating a polynomial of degree less than $n=2^m$ over an affine space $\theta+W_m$, the Gao--Mateer algorithm requires $\frac{1}{4}n(\log_2 n)^2+\frac{3}{4}n\log_2 n$ additions and $\frac{3}{2}n\log_2 n-n+1$ multiplications, respectively~\cite{BSG2026}. Hence, by the cost analysis in Section~\ref{Sec:costNewAFFT}, Algorithm~\ref{Algo:general-basis} and the Gao--Mateer algorithm perform the same number of additions, whereas the Gao--Mateer algorithm requires fewer multiplications. We now turn to the case of a partial Cantor special basis in $\mathbb{F}_{2^k}$.

Suppose that $f(x) \in \mathbb{F}_{2^k}[x]$ be a polynomial of degree less than $n=2^m$, and we want to evaluate $f(x)$ over the affine subspace $$\theta+W_m=\theta+\langle \beta_0,\beta_1,\ldots, \beta_{m_1-1},\beta_{m_1},\ldots, \beta_{m_1+m_2-1} \rangle,$$ where $\set{\beta_0,\beta_1,\ldots, \beta_{m_1-1}}$ is a set of Cantor special basis in the field $\mathbb{F}_{2^k}$. In this case, we have $\mathbb{Z}_{W_{m_1}}(x)=S^{m_1}(x)$. Consequently, the Taylor expansion of $f(x)$ with respect to $Z_{W_{m_1}}(x)$ incurs no multiplication cost.

In the following theorem, we determine the multiplication cost of Algorithm~\ref{Algo:general-basis} when the first $m_1$ element of the ordered basis of $W_m$ form a Cantor special basis.

\begin{theorem}\label{Th:cost-mult-algo1-partial}
    Let $m=m_1+m_2$, where $m_2\geq1$, and let $\theta+W_m =\theta+\langle\beta_0,\ldots,\beta_{m-1}\rangle \subseteq\mathbb{F}_{2^k}$ be an $m$-dimensional affine subspace. Suppose that the first $m_1$ basis elements form a Cantor special basis of maximum dimension admitted by $\mathbb{F}_{2^k}$. Then Algorithm~\ref{Algo:general-basis}, with split $m=m_1+m_2$, requires
    \[
    2^{m-2}\bigl(m_2^2+3m_2+2m_1\bigr)
    \]
    finite field multiplications to evaluate a polynomial of degree less than $2^m$ over $\theta+W_m$. Moreover, Algorithm~\ref{Algo:general-basis} requires fewer multiplications than the first Gao--Mateer AFFT if and only if $m_2^2-3m_2+4 \leq 4 m_1$.
\end{theorem}

\begin{proof}
Since the first $m_1$ basis elements form a Cantor special basis, the corresponding vanishing polynomial has coefficients in $\mathbb{F}_2$. Hence, the Taylor expansion on the $m_1$-side requires no finite field multiplications. Hence, the $A_{TE}(m_1)$ term in \eqref{Eqn:TE-cost} vanishes, and the multiplication cost of the Taylor expansion step in Algorithm~\ref{Algo:general-basis} satisfies
\[
A_{TE}(m)=2^{m_1}A_{TE}(m_2).
\]
It follows that the total multiplication cost of the Taylor expansion step is
\[
2^{m_1}\cdot \frac{1}{4}2^{m_2}(m_2^2-m_2)
=
\frac{1}{4}2^m(m_2^2-m_2).
\]

We now derive the multiplication performed by the recursive sub-FFTs. The $2^{m_1}$ column sub-AFFTs have dimension $m_2$ and are evaluated over a general basis. Since each requires $2^{m_2}m_2$ multiplications, their total cost is $2^{m_1}(2^{m_2}m_2)=2^m m_2$. The $2^{m_2}$ row sub-AFFTs have dimension $m_1$ and are evaluated over the Cantor special basis. Since each requires $\frac{1}{2} 2^{m_1}m_1$ multiplications, their total cost is $2^{m_2}\left(\frac{1}{2} 2^{m_1}m_1\right)= 2^{m-1} m_1$.

Therefore, the total multiplication cost of Algorithm~\ref{Algo:general-basis} is
\[
\frac{1}{4}2^m(m_2^2-m_2)+2^m m_2+ 2^{m-1} m_1
=
2^{m-2}\bigl(m_2^2+3m_2+2m_1\bigr).
\]

We now determine for which values of $m_1$ the multiplication cost of Algorithm~\ref{Algo:general-basis} is no greater than that of the Gao--Mateer algorithm. For simplicity, we use the expression
\[
\frac{3}{2}2^m m - 2^m = 2^{m-2}(6m-4)
\]
for the multiplication cost of the Gao--Mateer algorithm. Therefore,
\begin{equation*}
    \begin{aligned}
        & 2^{m-2}(m_2^2+3m_2+2m_1) \leq 2^{m-2}(6m-4) \\
        \iff & m_2^2+3m_2+2m_1 \leq 6m-4 \\
        \iff & m_2^2-3m_2+4  \leq 4 m_1.
    \end{aligned}
\end{equation*}

This completes the proof.
\end{proof}

% Therefore, the multiplication cost of the Taylor expansion step in Algorithm~\ref{Algo:general-basis} is reduced in this setting. More precisely, the Taylor expansion on the $m_1$-side incurs no multiplication cost, since division by $h(x)$ or its powers involves only coefficients in $\{0,1\}$. 

% Thus, for each fixed $m_1$, the inequality above yields an admissible range for $m_2$, and consequently for $m=m_1+m_2$. Table~\ref{Tab:Algo-1_partial-cantor} lists these ranges for the values of $m_1$ arising from the example fields under consideration.

\noindent Hence, for each fixed $m_1$, the condition $m_2^2-3m_2+4\leq 4m_1$ determines the range of $m_2$ for which Algorithm~\ref{Algo:general-basis} has less multiplication cost than the Gao--Mateer count. We next determine the addition cost of Algorithm~\ref{Algo:general-basis} under the same partial Cantor special basis assumption.

\begin{theorem}\label{Th:add-cost-Algo1-partial}
    Let $m=m_1+m_2$, where $m_2\geq1$, and let $\theta+W_m =\theta+\langle\beta_0,\ldots,\beta_{m-1}\rangle \subseteq\mathbb{F}_{2^k}$ be an $m$-dimensional affine subspace. Suppose that the first $m_1$ basis elements form a Cantor special basis of maximum dimension admitted by $\mathbb{F}_{2^k}$. Then Algorithm~\ref{Algo:general-basis}, with split $m=m_1+m_2$, requires
    \[
    2^{m-2}\left( m_2^2 + 5m_2 +4m_1 + m_1 \log_2 m_1  \right)
    \]
    finite field additions to evaluate a polynomial of degree less than $2^m$ over $\theta+W_m$.
\end{theorem}

\begin{proof}
    Since $m_1$ is the maximum dimension of a Cantor special basis admitted by $\mathbb{F}_{2^k}$, it is a power of two. Moreover, the first $m_1$ elements of the ordered basis generate a subspace $W_{m_1}$ whose vanishing polynomial is 
    \[
    Z_{W_{m_1}}(x)= S^{m_1}(x)=x^{2^{m_1}}+x.
    \]
    Thus, the first Taylor expansion in Algorithm~\ref{Algo:general-basis} is performed with respect to a binomial having coefficients in $\mathbb{F}_2$. Consequently, each coefficient update requires one addition and no finite field multiplication. More specifically, the Taylor expansion of a polynomial of degree less than $2^m$ with respect to $S^{m_1}(x)$ requires
    \[\frac{1}{2} 2^m \log(2^m / 2^{m_1})= 2^{m-1} m_2\]
    additions. Also, since $W_{m_1}$ is generated by a full Cantor special basis, the $2^{m_2}$ row sub-AFFTs of dimension $m_1$ recursively enjoy the same addition savings established in Theorem~\ref{Theorem:cost-Algo-fixed-m1}, contributing
    \[
    2^{m_2}\cdot\frac{1}{4}2^{m_1}m_1\log_2m_1= 2^{m-2}m_1\log_2m_1
    \]
    additions in total. Whereas, the $2^{m_1}$ column sub-AFFTs of dimension $m_2$ contribute the generic Taylor expansion step. Therefore, by Theorem~\ref{Th:costAlgo1}, the total addition cost due to Taylor expansion of the column sAFFTs is given by
    \[
    2^{m_1}\cdot\frac{1}{4}2^{m_2}(m_2^2-m_2)= 2^{m-2}(m_2^2-m_2).
    \]

    Finally, the additions performed by the recursive sub-FFTs are independent of the dimension split and contributing total $2^mm$ additions. Adding the four contributions, the total number of additions is given by
    \begin{equation*}
        \begin{aligned}
            & 2^{m-1} m_2 + 2^{m-2}m_1\log_2m_1 + 2^{m-2}(m_2^2-m_2) + 2^m m\\
            =&  2^{m-2}\left( m_2^2+5m_2 +4m_1 +m_1\log_2m_1 \right).
        \end{aligned}
    \end{equation*}

This proves the stated addition count.
\end{proof}

\begin{remark}\label{Remark:add-Algo1-GM}
    Comparing against the Gao--Mateer addition cost $2^{m-2}(m^2+3m)$ (Table~\ref{Table:AFFT_cost-cmp}) and substituting $m=m_1+m_2$, this algorithm requires no more additions than the Gao--Mateer algorithm exactly when
    \[
    2m_2(1-m_1) \;\leq\; m_1(m_1-\log_2m_1-1).
    \]
    For every $m_1\geq 2$, the right-hand side is nonnegative while the left-hand side is negative for every $m_2\geq 1$, so this inequality holds unconditionally. Hence, whenever a Cantor special basis of dimension $m_1\geq 2$ is available, Algorithm~\ref{Algo:general-basis} requires strictly fewer additions than the Gao--Mateer algorithm for \emph{every} admissible $m_2$, not only within the multiplication-limited range above. Consequently, the multiplication condition $m_2^2-3m_2+4\leq 4 m_1$ is the sole binding constraint for Algorithm~\ref{Algo:general-basis} to dominate Gao--Mateer in \emph{both} additions and multiplications simultaneously.
\end{remark}

% \noindent Thus, for each fixed $m_1$, the multiplication inequality derived in Theorem~\ref{Th:cost-mult-algo1-partial} yields an admissible range for $m_2$, and consequently for $m=m_1+m_2$. Table~\ref{Tab:Algo-1_partial-cantor} lists these ranges for the values of $m_1$ arising from the example fields under consideration.

\begin{table}[htb]
\centering
\caption{Admissible parameter ranges for which Algorithm~\ref{Algo:general-basis}, by exploiting a partial Cantor special basis, requires fewer finite field multiplications and additions than the first Gao--Mateer AFFT. The ranges satisfy the simplified sufficient condition $m_2^2-3m_2+4\leq 4 m_1$, where $m=m_1+m_2$.}
\label{Tab:Algo-1_partial-cantor}
\renewcommand{\arraystretch}{1.2}
\begin{tabular}{|c|c|c|c|}
\hline
Finite Field & $m_1$ & Admissible range of $m_2$ & Admissible range of $m$ \\
\hline
$\mathbb{F}_{2^{10}}$ & 2  & $1 \leq m_2 \leq 4$ & $3 \leq m \leq 6$ \\
$\mathbb{F}_{2^{12}}$ & 4  & $1 \leq m_2 \leq 5$ & $5 \leq m \leq 9$ \\
$\mathbb{F}_{2^{24}}$ & 8  & $1 \leq m_2 \leq 7$ & $9 \leq m \leq 15$ \\
$\mathbb{F}_{2^{48}}$ & 16 & $1 \leq m_2 \leq 9$ & $17 \leq m \leq 25$ \\
\hline
\end{tabular}
\end{table}

\noindent Table~\ref{Tab:Algo-1_partial-cantor} shows that the admissible range generally expands as the dimension $m_1$ of the available Cantor special basis increases. For example, over $\mathbb{F}_{2^{10}}$, where $m_1=2$, Algorithm~\ref{Algo:general-basis} outperforms the first Gao--Mateer algorithm in both additions and multiplications for $3\leq m\leq 6$. Over $\mathbb{F}_{2^{48}}$, where $m_1=16$, the corresponding range is $17\leq m\leq 25$. Thus, when the available Cantor special basis does not span the full evaluation subspace, its partial structure yields a nontrivial range of AFFT dimensions for which Algorithm~\ref{Algo:general-basis} has lower arithmetic costs than the first Gao--Mateer algorithm.

% We emphasize that the admissible ranges in Table~\ref{Tab:Algo-1_partial-cantor} are derived solely from the comparison of multiplication costs. In the general-basis setting, both Algorithm~\ref{Algo:general-basis} and the Gao--Mateer algorithm have the same addition cost, whereas the Gao--Mateer algorithm requires fewer multiplications. However, when a partial Cantor special basis exists, Algorithm~\ref{Algo:general-basis} has lower multiplication cost than the Gao--Mateer algorithm for the admissible values of $m$ listed in Table~\ref{Tab:Algo-1_partial-cantor}, while also requiring fewer additions.

\section{Implementation and Benchmarking}\label{Sec:Implementation}
We implemented Algorithm~\ref{Algo:AFFT-Cantor-fixed-m1} in the C programming language. We refer to the implementation of this algorithm as the \emph{Dyadic AFFT} because of its hierarchical subdivision into subproblems whose sizes are powers of two. The source code is publicly available\footnote{\url{https://github.com/mtbadakhshan/additive-fft/C/dyadic}}.

Over a Cantor special basis, Algorithm~\ref{Algo:AFFT-Cantor-fixed-m1} and the LCH AFFT both require $\frac{1}{2}n\log_2 n$ finite-field multiplications, where $n=2^m$. However, unlike the LCH AFFT, which consists of separate basis-conversion and evaluation stages, our proposed Algorithm~\ref{Algo:AFFT-Cantor-fixed-m1} has a fully recursive structure with no separate stages. This structure provides better cache locality by design, thereby reducing memory-access overhead.

The memory overhead of the Dyadic AFFT is $O(1)$, as the algorithm operates in place without requiring auxiliary buffers. Each subproblem is laid out within the input array such that the distance between successive coefficients is determined by a run-time stride parameter. Moreover, all subproblems within a given column or row loop in Algorithm~\ref{Algo:AFFT-Cantor-fixed-m1} share the same Taylor-expansion pattern and stride, differing only in their starting offsets. This organization results in a regular and predictable memory-access pattern throughout the recursive computation. Consequently, once the working set of a subproblem fits within the cache, its descendant subproblems can be processed largely from cache, thereby reducing accesses to main memory and improving data locality.

\paragraph*{Benchmark}
We benchmark our implementation of Algorithm~\ref{Algo:AFFT-Cantor-fixed-m1} (The Dyadic AFFT) and compared its execution time with that of the LCH AFFT implementation~\cite{LCH-FFT2016} over a Cantor special basis. The LCH AFFT is a widely used method for evaluation over a Cantor special basis (e.g., \cite{LCH-Fast_Mult2018,LCH-Frobenius2018_2,LCH-Frobenius2018,BSG2026,chen2026HQC}). Both AFFT algorihtms require $\frac{1}{2}n\log_2 n$ finite-field multiplications, making Algorithm~\ref{Algo:AFFT-Cantor-fixed-m1} the most direct arithmetic comparison with LCH among the proposed algorithms in this paper. Both implementations operate over the binary extension field $\mathbb{F}_{2^{128}}$. To ensure a fair comparison, our implementation directly reuses the routines provided by the LCH implementation for field multiplication and for deriving the twiddle factors required in the base case (i.e., $m=1$).

\paragraph*{Experimental Setup} The experiments were conducted on two systems. Platform~A is equipped with an Intel Core i5-1250P processor operating at up to 4.4~GHz, 16~GB of LPDDR5-5200 memory, and Arch Linux (kernel 6.18.37-1-lts). Platform~B is equipped with an AMD Ryzen 9 9950X processor operating at up to 5.7~GHz, 64~GB of DDR5 memory, and Debian GNU/Linux~12 (kernel 6.12.12). 

All benchmarks were executed in a single thread. For each AFFT dimension~$m$, each algorithm was executed 300 times. The results are reported as the arithmetic mean running time over these iterations, together with the standard deviation. The AFFT length is $n=2^m$, and all running times are reported in milliseconds. 

\paragraph*{Benchmark Results}
Let $T_{\mathrm{LCH}}(m)$ and $T_{\mathrm{Algo~\ref{Algo:AFFT-Cantor-fixed-m1}}}(m)$ denote the mean running times of the LCH AFFT and the Dyadic AFFT, respectively, at dimension $m$. We define the speedup of the Dyadic AFFT relative to the LCH AFFT as $S_m =\frac{T_{\mathrm{LCH}}(m)}{T_{\mathrm{Algo~\ref{Algo:AFFT-Cantor-fixed-m1}}}(m)}$. Thus, $S_m>1$ indicates that the Algorithm~\ref{Algo:AFFT-Cantor-fixed-m1} is faster, whereas $S_m<1$ indicates that the LCH AFFT is faster.
Table~\ref{Table:benchmark-results} reports the mean running times of the LCH AFFT and Algorithm~\ref{Algo:AFFT-Cantor-fixed-m1} (Dyadic AFFT), together with
the corresponding speedups, on both platforms. Platform~A was tested for $9\leq m\leq28$, while Platform~B was tested for $9\leq m\leq 30$.

\begin{table*}[htb]
\centering
\caption{Benchmark comparison of the LCH AFFT and Algorithm~\ref{Algo:AFFT-Cantor-fixed-m1} across input sizes $m = 9,\dots,30$, evaluated on two hardware platforms. Reported values are mean running times with standard deviations, in milliseconds, averaged over 300 iterations. The speedup factor $S_m = T_{\mathrm{LCH}}/T_{\mathrm{Algo~\ref{Algo:AFFT-Cantor-fixed-m1}}}$ quantifies the performance advantage of Algorithm~\ref{Algo:AFFT-Cantor-fixed-m1} relative to the LCH AFFT at each size.}
\label{Table:benchmark-results}
\scriptsize
\setlength{\tabcolsep}{3.5pt}
\renewcommand{\arraystretch}{1.3}
\begin{tabular}{c|ccr|ccr}
\hline
&
\multicolumn{3}{c|}{Platform A} &
\multicolumn{3}{c}{Platform B} \\
$m$ &
LCH (ms) &
Algorithm~\ref{Algo:AFFT-Cantor-fixed-m1} (ms) &
$S_m$ &
LCH (ms) &
Algorithm~\ref{Algo:AFFT-Cantor-fixed-m1} (ms) &
$S_m$ \\
\hline
 9 & $0.0068 \pm 0.0005$ & $0.0053 \pm 0.0003$ & 1.283
   & $0.0058 \pm 0.0002$ & $0.0049 \pm 0.0001$ & 1.184 \\
10 & $0.0144 \pm 0.0004$ & $0.0115 \pm 0.0003$ & 1.252
   & $0.0126 \pm 0.0002$ & $0.0109 \pm 0.0001$ & 1.156 \\
11 & $0.0314 \pm 0.0006$ & $0.0257 \pm 0.0005$ & 1.222
   & $0.0271 \pm 0.0003$ & $0.0244 \pm 0.0002$ & 1.111 \\
12 & $0.0689 \pm 0.0008$ & $0.0580 \pm 0.0009$ & 1.188
   & $0.0590 \pm 0.0004$ & $0.0546 \pm 0.0002$ & 1.081 \\
13 & $0.1487 \pm 0.0012$ & $0.1283 \pm 0.0018$ & 1.159
   & $0.1269 \pm 0.0013$ & $0.1209 \pm 0.0003$ & 1.050 \\
14 & $0.3209 \pm 0.0049$ & $0.2861 \pm 0.0059$ & 1.122
   & $0.2724 \pm 0.0027$ & $0.2630 \pm 0.0005$ & 1.036 \\
15 & $0.6852 \pm 0.0074$ & $0.6020 \pm 0.0017$ & 1.138
   & $0.6169 \pm 0.0052$ & $0.6093 \pm 0.0023$ & 1.012 \\
16 & $1.5652 \pm 0.0057$ & $1.4351 \pm 0.0051$ & 1.091
   & $1.3188 \pm 0.0126$ & $1.3108 \pm 0.0024$ & 1.006 \\
17 & $3.7949 \pm 0.0088$ & $3.5196 \pm 0.0131$ & 1.078
   & $2.8022 \pm 0.0253$ & $2.8355 \pm 0.0058$ & 0.988 \\
18 & $8.1758 \pm 0.0465$ & $8.0434 \pm 0.0496$ & 1.016
   & $5.7276 \pm 0.0490$ & $5.8051 \pm 0.0064$ & 0.987 \\
19 & $18.9797 \pm 0.4522$ & $16.9004 \pm 0.2734$ & 1.123
   & $11.9291 \pm 0.0674$ & $12.2396 \pm 0.0178$ & 0.975 \\
20 & $46.1808 \pm 1.1991$ & $40.7112 \pm 0.7746$ & 1.134
   & $25.3817 \pm 0.0932$ & $26.1196 \pm 0.0254$ & 0.972 \\
21 & $183.9020 \pm 48.9503$ & $142.2784 \pm 34.6881$ & 1.293
   & $56.2517 \pm 0.2424$ & $58.4476 \pm 0.1072$ & 0.962 \\
22 & $429.4120 \pm 18.4772$ & $336.0054 \pm 9.5468$ & 1.278
   & $136.0706 \pm 0.4392$ & $130.3026 \pm 0.2912$ & 1.044 \\
23 & $916.0992 \pm 30.9932$ & $676.2897 \pm 16.7546$ & 1.355
   & $297.3184 \pm 0.7391$ & $284.1438 \pm 0.4391$ & 1.046 \\
24 & $2117.3995 \pm 24.1954$ & $1491.1270 \pm 17.4921$ & 1.420
   & $644.5942 \pm 1.1382$ & $611.6455 \pm 0.8414$ & 1.054 \\
25 & $4672.1950 \pm 42.1372$ & $3207.3658 \pm 50.6881$ & 1.457
   & $1358.2545 \pm 1.7262$ & $1301.5135 \pm 1.0899$ & 1.044 \\
26 & $10102.0037 \pm 48.7173$ & $6837.3803 \pm 87.0749$ & 1.477
   & $2861.3927 \pm 2.5597$ & $2734.3270 \pm 1.4908$ & 1.046 \\
27 & $21627.2300 \pm 69.4215$ & $14943.3218 \pm 123.8313$ & 1.447
   & $6053.9774 \pm 83.5571$ & $5871.4243 \pm 89.2968$ & 1.031 \\
28 & $46802.5562 \pm 293.2676$ & $33453.8079 \pm 390.4190$ & 1.399
   & $12825.6158 \pm 194.5572$ & $12417.0028 \pm 210.6376$ & 1.033 \\
29 & \multicolumn{3}{c|}{--}
   & $26757.3777 \pm 304.3647$ & $26187.5678 \pm 329.0410$ & 1.022 \\
30 & \multicolumn{3}{c|}{--}
   & $56927.2574 \pm 743.1407$ & $55624.5156 \pm 806.8231$ & 1.023 \\
\hline

\end{tabular}
\end{table*}

On Platform~A, Algorithm~\ref{Algo:AFFT-Cantor-fixed-m1} is uniformly faster than the LCH AFFT across all tested dimensions. This improvement is more significant on Platform~A due to its slower memory subsystem, which amplifies the benefits of the lower memory-access overhead of our algorithm. The observed speedups range from $1.02\times$ at $m=18$ to $1.48\times$ at $m=26$. At the largest tested dimension, $m=28$, our algorithm reduces the running time from $46.803$ to $33.454$ seconds, corresponding to a speedup of $1.40\times$ and an absolute reduction of $13.349$ seconds.

On Platform~B, the Algorithm~\ref{Algo:AFFT-Cantor-fixed-m1} is faster for $9\leq m\leq16$ and $22\leq m\leq30$, whereas the LCH AFFT is slightly faster for $17\leq m\leq21$. The largest LCH advantage is observed at $m=21$, where its running time is approximately $3.8\%$ lower. The Algorithm~\ref{Algo:AFFT-Cantor-fixed-m1} regains the advantage at $m=22$ and remains faster through $m=30$. At $m=30$, it reduces the mean running time from $56.927$ to $55.625$ seconds, corresponding to a speedup of $1.02$.
Overall, Algorithm~\ref{Algo:AFFT-Cantor-fixed-m1} is faster in 37 of the 42 measured cases, including all 20 cases on Platform~A and 17 of the 22 cases on Platform~B. The results therefore show that the proposed algorithm generally achieves lower running times than the LCH
AFFT, although the relative performance depends on the transform dimension and the target platform.

\paragraph*{Discussion of Other Proposed Algorithms} We implemented Algorithm~\ref{Algo:AFFT-Cantor-fixed-m1} for the experimental comparison. Although Algorithm~\ref{Algo:AFFT-Cantor-any-m1} provides the more general framework, supporting any ordered basis and any decomposition $m=m_1+m_2$, over a general ordered basis it has the same addition count as the first Gao--Mateer algorithm~\cite{Gao2010FFT} but requires more finite field multiplications. Nevertheless, Algorithm~\ref{Algo:general-basis} offers two important advantages. First, its arithmetic complexity is independent of the dimension split (Theorem~\ref{Th:costAlgo1}), allowing $m_1$ and $m_2$ to be selected according to the parallelism, cache hierarchy, memory capacity, and data-access characteristics of the target device. Second, when a partial Cantor special basis is available, Algorithm~1 can exploit this structure (see Section~\ref{Sec:partial-GMvAlgo1}) to reduce both additions and multiplications, whereas the first Gao--Mateer algorithm cannot generally obtain the same benefit because its evaluation basis changes during the recursion.

Algorithm~\ref{Algo:AFFT-Cantor-any-m1} specializes this matrix-decomposition framework to a Cantor special basis and retains the freedom to choose an arbitrary split $m=m_1+m_2$. Thus, it also permits the dimensions of the column and row sub-AFFTs to be selected according to device-dependent considerations such as parallelism, cache utilization, and memory organization. However, an arbitrary split may increase addition count in the Taylor expansion stage. Algorithm~\ref{Algo:AFFT-Cantor-fixed-m1} instead sets $m_1= 2^{\floor{\log_2(m-1)}}$ and $m_2=m-m_1$ at each recursive step. Since $m_1$ is a power of two, the corresponding subspace polynomial has a binomial form, so each division step in the Taylor expansion requires only one addition. This algorithm therefore trades splitting flexibility for a lower addition count determined by the binary representation of (see Theorem~\ref{Theorem:cost-Algo-fixed-m1}).

%%%%%%%%%%%%%%%%%%%%%%
%%%LCH
%%%%%%%%%%%%%%%%%%%%%%
% \input{LCH-butterfly}

\section{Concluding Remarks}\label{Sec:conclusion}
In this paper, we developed a matrix-decomposition framework for additive fast Fourier transforms (AFFTs) over binary extension fields. Motivated by the Bailey's four-step FFT~\cite{Bailey1990}, the proposed framework uses Taylor expansion with respect to a subspace vanishing polynomial to arrange the coefficients of the input polynomial into a matrix and decompose the original evaluation problem into independent column and row sub-AFFTs. This formulation applies to any ordered basis and any decomposition $m=m_1+m_2$. The resulting general-basis algorithm requires $\frac{1}{4}n(\log_2 n)^2+\frac{3}{4}n\log_2 n$ additions and the same number of multiplications, independently of the chosen dimension split. The Bailey-style matrix decomposition provides structural and practical advantages. In particular, this invariance permits the split to be selected according to implementation requirements, such as parallelism and memory organization, without changing the arithmetic complexity.

We then specialized the framework to evaluation subspaces admitting a Cantor special basis. In this setting, the relevant subspace vanishing polynomials have binary coefficients, allowing the Taylor expansion stage to be performed without finite field multiplications. We first presented an algorithm supporting an arbitrary decomposition $m=m_1+m_2$, which provides flexibility in choosing the dimensions of the column and row sub-AFFTs. We then derived a recursive algorithm in which $m_1$ is chosen as the largest power of two smaller than $m$. This choice preserves the binomial form of the relevant subspace vanishing polynomial throughout the recursion and minimizes the addition cost of the Taylor expansion stage. The algorithm requires exactly $\frac{1}{2}n\log_2 n$ finite field multiplications, together with a closed-form number of additions determined by the binary representation of $m$. Our implementation and benchmark results show that Algorithm~\ref{Algo:AFFT-Cantor-fixed-m1} outperforms the LCH AFFT~\cite{LCH-basis2014,LCH-FFT2016} over a Cantor special basis in 37 of the 42 tested configurations, and at every tested dimension on one of the two platforms.

We further generalized the butterfly phase of the LCH AFFT in the novel polynomial basis to an arbitrary decomposition $m = m_1 + m_2$, showing in Theorem~\ref{Th:costGenLCH} that its cost is invariant under the split, at exactly $n \log_2 n$ additions and $\tfrac{1}{2} n \log_2 n$ multiplications.
Finally, we formalized the notion of a partial Cantor special basis, in which only a prefix of dimension $\ell \le m$ of the ordered basis satisfies the Cantor recursion. We showed that this partial structure reduces the arithmetic costs of both the von zur Gathen--Gerhard AFFT~\cite{zurGathenFFT} and our general-basis algorithm (Algorithm~\ref{Algo:general-basis}), whereas the first Gao--Mateer algorithm~\cite{Gao2010FFT} cannot exploit the same advantage because its evaluation basis changes during the recursion. The advantage of the von zur Gathen--Gerhard AFFT is limited to $m\in\{\ell+1,\ell+2\}$ by the condition $\ell\ge m-2$ in Theorem~\ref{Th:cost-mult-GG-partial}. In contrast, setting $m_1=\ell$ allows Algorithm~\ref{Algo:general-basis} to outperform the first Gao--Mateer AFFT over substantially wider ranges. For example, over $\mathbb{F}_{2^{48}}$ with $m_1=16$, it outperforms the first Gao--Mateer AFFT for $17\le m\le 25$ (Table~\ref{Tab:Algo-1_partial-cantor}).

This work suggests several promising directions for future research. The first problem is to extend the proposed matrix-decomposition framework to truncated AFFTs. This would require determining how to restrict the Taylor coefficient matrix and the row and column sub-AFFTs to prescribed subsets of evaluation points, as well as deriving arithmetic complexity bounds that depend on the number and structure of the requested evaluations. 
The second problem is to develop Frobenius variants~\cite{LCH-Frobenius2018} of Algorithms~\ref{Algo:AFFT-Cantor-any-m1} and~\ref{Algo:AFFT-Cantor-fixed-m1}. For polynomials over $\mathbb{F}_2$, the relation $f(a^2)=f(a)^2$  allows the algorithms to evaluate the polynomial only at representatives of the Frobenius orbits and to recover the remaining values by repeated squaring. This requires determining how such orbit representatives relate to the Taylor coefficient matrix and to the induced row and column sub-AFFTs. 
The third  problem is to reduce the $\mathcal{O}(n(\log_2 n)^2)$ arithmetic complexity of Algorithm~\ref{Algo:general-basis}. This could be achieved through a more efficient Taylor-expansion stage or an alternative algebraic decomposition, while preserving the row--column structure for arbitrary ordered bases. Establishing lower bounds on the required numbers of additions and multiplications over a Cantor special basis would further clarify whether the operation counts achieved by Algorithm~\ref{Algo:AFFT-Cantor-fixed-m1} can be reduced.

% Based on this work, several promising directions remain for future research. The split-invariant arithmetic complexity of the proposed Algorithm~\ref{Algo:general-basis} makes it natural to investigate architecture-aware choices of $m_1$ and $m_2$, since different decompositions may provide different trade-offs in parallelism, cache locality, and memory traffic. Algorithm~\ref{Algo:AFFT-Cantor-any-m1} is another promising candidate for implementation. It operates directly on polynomials represented in the standard monomial basis and avoids the separate basis-conversion stage required by LCH. Its freedom to choose an arbitrary split $m=m_1+m_2$ could be used to match the subproblem dimensions to the target cache hierarchy, potentially reducing cache misses and memory traffic. Such architecture-aware choices may incur more additions than choosing $m_1$ recursively as in Algorithm~\ref{Algo:AFFT-Cantor-fixed-m1}. Implementations would therefore help determine whether improved memory-access behavior can compensate for any additional arithmetic cost in practice. It would also be useful to investigate efficient constructions of partial Cantor special bases and optimize the interaction between the available Cantor prefix and the dimension split. Finally, evaluating the proposed AFFTs in applications such as polynomial arithmetic, error-correcting codes, and proof systems over binary extension fields would reveal how their arithmetic costs and memory-access characteristics affect end-to-end performance.

% \medskip

% \bibliographystyle{plain}
\bibliographystyle{IEEEtran}
\bibliography{Refined_Ref,biblio}

@misc{STARK2018,
      author = {Eli Ben-Sasson and Iddo Bentov and Yinon Horesh and Michael Riabzev},
      title = {Scalable, transparent, and post-quantum secure computational integrity},
      howpublished = {Cryptology {ePrint} Archive, Paper 2018/046},
      year = {2018},
      url = {https://eprint.iacr.org/2018/046}
}

@article{CooleyTukey1965,
    ISSN = {00255718, 10886842},
    URL = {http://www.jstor.org/stable/2003354},
    author = {James W. Cooley and John W. Tukey},
    journal = {Mathematics of Computation},
    number = {90},
    pages = {297--301},
    publisher = {American Mathematical Society},
    title = {An Algorithm for the {M}achine {C}alculation of {C}omplex {F}ourier {S}eries},
    volume = {19},
    year = {1965}
}

@ARTICLE{Gao2010FFT,
  author={Gao, Shuhong and Mateer, Todd},
  journal={IEEE Transactions on Information Theory}, 
  title={Additive {F}ast {F}ourier {T}ransforms {O}ver {F}inite {F}ields}, 
  year={2010},
  volume={56},
  number={12},
  pages={6265-6272},
  doi={10.1109/TIT.2010.2079016},
  url     = {https://doi.org/10.1109/TIT.2010.2079016}
}

@ARTICLE{WangZhu1988,
    author={Wang, Yao and Zhu, Xuelong},
    journal={IEEE Journal on Selected Areas in Communications}, 
    title={A fast algorithm for the {F}ourier transform over finite fields and its {VLSI} implementation}, 
    year={1988},
    volume={6},
    number={3},
    pages={572-577},
    doi     = {10.1109/49.1926},
    url     = {https://ieeexplore.ieee.org/document/1926}
}

@article{Cantor1989FFT,
  title = {On arithmetical algorithms over finite fields},
  journal = {Journal of Combinatorial Theory, Series A},
  volume = {50},
  number = {2},
  pages = {285-300},
  year = {1989},
  issn = {0097-3165},
  doi = {https://doi.org/10.1016/0097-3165(89)90020-4},
  url = {https://www.sciencedirect.com/science/article/pii/0097316589900204},
  author = {David G. Cantor}
}

@inproceedings{zurGathenFFT,
  author = {Gathen, Joachim von zur and Gerhard, J\"{u}rgen},
  title = {Arithmetic and factorization of polynomial over $\mathbb{F}_2$ (extended abstract)},
  year = {1996},
  isbn = {0897917960},
  publisher = {Association for Computing Machinery},
  address = {New York, NY, USA},
  url = {https://doi.org/10.1145/236869.236882},
  doi = {10.1145/236869.236882},
  booktitle = {Proceedings of the 1996 International Symposium on Symbolic and Algebraic Computation},
  pages = {1–9},
  numpages = {9},
  location = {Zurich, Switzerland},
  series = {ISSAC 1996}
}

@InProceedings{Aurora2019,
    author={Ben-Sasson, Eli
    and Chiesa, Alessandro
    and Riabzev, Michael
    and Spooner, Nicholas
    and Virza, Madars
    and Ward, Nicholas P.},
    editor={Ishai, Yuval
    and Rijmen, Vincent},
    title={Aurora: {T}ransparent {S}uccinct {A}rguments for {R1CS}},
    booktitle={Advances in Cryptology -- EUROCRYPT 2019},
    year={2019},
    publisher={Springer International Publishing},
    address={Cham},
    doi  = {10.1007/978-3-030-17653-2_4},
    url       = {https://link.springer.com/chapter/10.1007/978-3-030-17653-2_4},
    pages={103-128}
}

@article{Polaris2022,
  title={{Polaris: Transparent Succinct Zero-Knowledge Arguments for {R1CS} with Efficient Verifier}},
  author={Fu, Shihui and Gong, Guang},
  journal={Proceedings on Privacy Enhancing Technologies},
  year={2022},
  url          = {https://doi.org/10.2478/popets-2022-0027},
  doi          = {10.2478/POPETS-2022-0027}
}

@Inbook{Gao2003,
author="Gao, Shuhong",
editor="Bhargava, Vijay K.
and Poor, H. Vincent
and Tarokh, Vahid
and Yoon, Seokho",
title="A {N}ew {A}lgorithm for {D}ecoding {Reed-S}olomon {C}odes",
bookTitle="Communications, Information and Network Security",
year="2003",
publisher="Springer US",
address="Boston, MA",
pages="55--68",
isbn="978-1-4757-3789-9",
doi="10.1007/978-1-4757-3789-9_5",
url="https://doi.org/10.1007/978-1-4757-3789-9_5"
}

@article{LCH-basis2014,
    title={Novel {P}olynomial {B}asis and {I}ts {A}pplication to {R}eed-{S}olomon {E}rasure {C}odes},
    author={Sian-Jheng Lin and Wei-Ho Chung and Yunghsiang Sam Han},
    journal={2014 IEEE 55th Annual Symposium on Foundations of Computer Science},
    year={2014},
    pages={316-325},
    url = {https://doi.org/10.1109/FOCS.2014.41},
    doi = {10.1109/FOCS.2014.41}
}

@ARTICLE{LCH-conv2016,
    author={Sian-Jheng Lin and Tareq Y. {Al-Naffouri} and Yunghsiang Sam Han and Wei-Ho Chung},
    journal={IEEE Transactions on Information Theory}, 
    title={Novel {P}olynomial {B}asis {W}ith {F}ast {F}ourier {T}ransform and {I}ts {A}pplication to {Reed–S}olomon {E}rasure Codes}, 
    year={2016},
    volume={62},
    number={11},
    pages={6284-6299},
    doi     = {10.1109/TIT.2016.2608892},
    url     = {https://ieeexplore.ieee.org/abstract/document/7565465/}
}

@ARTICLE{LCH-FFT2016,
    author={Sian-Jheng Lin and Tareq Y. {Al-Naffouri} and Yunghsiang Sam Han},
    journal={IEEE Transactions on Information Theory}, 
    title={{FFT} {A}lgorithm for {B}inary {E}xtension {F}inite {F}ields and {I}ts {A}pplication to {R}eed–{S}olomon {C}odes}, 
    year={2016},
    volume={62},
    number={10},
    pages={5343-5358},
    doi={10.1109/TIT.2016.2600417},
    url     = {https://ieeexplore.ieee.org/document/7543456/}
}

@article{COXON2021,
    title = {Fast transforms over finite fields of characteristic two},
    journal = {Journal of Symbolic Computation},
    volume = {104},
    pages = {824-854},
    year = {2021},
    issn = {0747-7171},
    doi = {https://doi.org/10.1016/j.jsc.2020.10.002},
    url = {https://www.sciencedirect.com/science/article/pii/S0747717120301127},
    author = {Nicholas Coxon}
}

@InProceedings{BernsteinChou2014,
    author="Bernstein, Daniel J.
    and Chou, Tung",
    editor="Joux, Antoine
    and Youssef, Amr",
    title="Faster {B}inary-{F}ield {M}ultiplication and {F}aster {B}inary-{F}ield {MACs}",
    booktitle="Selected Areas in Cryptography -- SAC 2014",
    year="2014",
    publisher="Springer International Publishing",
    address="Cham",
    pages="92--111",
    doi       = {10.1007/978-3-319-13051-4_6},
    url       = {https://doi.org/10.1007/978-3-319-13051-4_6}
}

@Article{Ames2017Ligero,
author={Ames, Scott
and Hazay, Carmit
and Ishai, Yuval
and Venkitasubramaniam, Muthuramakrishnan},
title={Ligero: {L}ightweight {S}ublinear {A}rguments {W}ithout a {T}rusted {S}etup},
journal={Designs, Codes and Cryptography},
year={2023},
month={Nov},
day={01},
volume={91},
number={11},
pages={3379-3424},
issn={1573-7586},
doi={10.1007/s10623-023-01222-8},
url={https://doi.org/10.1007/s10623-023-01222-8}
}

@InProceedings{Chiesa2020Fractal,
author="Chiesa, Alessandro
and Ojha, Dev
and Spooner, Nicholas",
editor="Canteaut, Anne
and Ishai, Yuval",
title="{Fractal: Post-quantum and Transparent Recursive Proofs from Holography}",
booktitle="Advances in Cryptology -- EUROCRYPT 2020",
year="2020",
publisher="Springer International Publishing",
address="Cham",
pages="769--793",
doi       = {10.1007/978-3-030-45721-1_27},
url       = {https://link.springer.com/chapter/10.1007/978-3-030-45721-1_27},
isbn="978-3-030-45721-1"
}

@misc{LCH-Fast_Mult2018,
    title={Faster {M}ultiplication for {L}ong {B}inary {P}olynomials}, 
    author={Ming-Shing Chen and Chen-Mou Cheng and Po-Chun Kuo and Wen-Ding Li and Bo-Yin Yang},
    year={2018},
    archivePrefix={arXiv},
    howpublished  = {arXiv:~1708.09746},
    primaryClass={cs.SC},
    url={https://arxiv.org/abs/1708.09746}
    
}

@misc{LCH-Frobenius2018_2,
  author = {{M.-S Chen and C.-M Cheng and P.-C Kuo and W.-D Li and B.-Y Yang}},
  title  = {Multiplying {B}oolean Polynomials with {F}robenius Partitions in Additive Fast {F}ourier Transform},
  year   = {2018},
  note   = {arXiv:1803.11301},
  url    = {https://arxiv.org/abs/1803.11301}
}

@inproceedings{LCH-Frobenius2018,
    author = {Li, Wen-Ding and Chen, Ming-Shing and Kuo, Po-Chun and Cheng, Chen-Mou and Yang, Bo-Yin},
    title = {Frobenius {A}dditive {F}ast {F}ourier {T}ransform},
    year = {2018},
    url = {https://doi.org/10.1145/3208976.3208998},
    doi = {10.1145/3208976.3208998},
    booktitle = {Proceedings of the 2018 ACM International Symposium on Symbolic and Algebraic Computation},
    pages = {263–270},
    numpages = {8},
    series = {ISSAC 2018}
}

@InProceedings{ECFFT1_2023,
author = {Eli Ben-Sasson and Dan Carmon and Swastik Kopparty and David Levit},
title = {Elliptic {C}urve {F}ast {F}ourier {T}ransform ({ECFFT}) {Part I}: {L}ow-degree {E}xtension in {T}ime $O(n \log n)$ over all Finite Fields},
booktitle = {Proceedings of the 2023 Annual ACM-SIAM Symposium on Discrete Algorithms (SODA)},
year={2023},
pages = {700-737},
doi = {10.1137/1.9781611977554.ch30},
URL = {https://epubs.siam.org/doi/abs/10.1137/1.9781611977554.ch30},
eprint = {https://epubs.siam.org/doi/pdf/10.1137/1.9781611977554.ch30},
}

@misc{CircleFFT2024,
    author = {Ulrich Haböck and David Levit and Shahar Papini},
    title = {Circle {STARKs}},
    howpublished = {Cryptology {ePrint} Archive, Paper 2024/278},
    year = {2024},
    url = {https://eprint.iacr.org/2024/278}
}

@inproceedings{GFFT2024,
  author    = {Songsong Li and Chaoping Xing},
  title     = {Fast {F}ourier Transform via Automorphism Groups of Rational Function Fields},
  booktitle = {Proceedings of the 2024 Annual ACM-SIAM Symposium on Discrete Algorithms ({SODA})},
  pages     = {3836--3859},
  year      = {2024},
  publisher = {Society for Industrial and Applied Mathematics ({SIAM})},
  doi       = {10.1137/1.9781611977912.135},
  url       = {https://doi.org/10.1137/1.9781611977912.135}
}

@InProceedings{BSG2026,
	author="Badakhshan, Mohammadtaghi
	and Samanta, Susanta
	and Gong, Guang",
	editor="Boura, Christina
	and Mashatan, Atefeh
	and Miri, Ali",
	title="Accelerating Post-quantum Secure {zkSNARKs} by Optimizing Additive {FFT}",
	booktitle="Selected Areas in Cryptography -- SAC 2025",
	year="2026",
	publisher="Springer Nature Switzerland",
	address="Cham",
	pages="339--368",
	isbn="978-3-032-10536-3",
    doi       = {10.1007/978-3-032-10536-3_13},
    url       = {https://link.springer.com/chapter/10.1007/978-3-032-10536-3_13}
}

@article{Bailey1990,
  author  = {David H. Bailey},
  title   = {{FFTs} in External or Hierarchical Memory},
  journal = {The Journal of Supercomputing},
  volume  = {4},
  number  = {1},
  pages   = {23--35},
  year    = {1990},
  doi     = {10.1007/BF00162341},
  url     = {https://doi.org/10.1007/BF00162341}
}

@article{chen2026HQC,
    author  = {Chen, Ming-Shing and Chiu, Chun-Ming and Peng, Chun-Tao and Yang, Bo-Yin},
    title   = {Accelerating {HQC} with {A}dditive {FFT}},
    journal = {IACR Transactions on Cryptographic Hardware and Embedded Systems},
    year    = {2026},
    volume  = {2026},
    number  = {2},
    pages   = {520--544},
    month   = apr,
    doi     = {10.46586/tches.v2026.i2.520-544},
    url     = {https://tches.iacr.org/index.php/TCHES/article/view/12898}
}

% \medskip

\appendices

\section{Von zur Gathen--Gerhard and Cantor AFFT Algorithms}\label{Sec:GG-Cantor-AFFT}
This section reviews the classical additive FFTs of Cantor~\cite{Cantor1989FFT} and von zur Gathen and Gerhard~\cite{zurGathenFFT}, focusing on the features used in the partial Cantor basis analysis of Section~\ref{Sec:partial-Cantor}. Cantor introduced the additive FFT in 1989 for evaluation subspaces generated by a Cantor special basis. In 1996, von zur Gathen and Gerhard generalized Cantor's approach to general ordered bases, at the cost of additional finite field operations.

Let $f(x) \in \mathbb{F}_{2^k}[x]$ be a polynomial of degree less than $n=2^m$ and we want to evaluate $f(x)$ over the affine subspace $\theta+W_m=\theta+\langle \beta_0,\beta_1,\ldots,\beta_{m-1} \rangle$. The evaluation of $f(x)$ using the Von zur Gathen--Gerhard algorithm proceeds as follows: first, compute two polynomials $f_0(x)$ and $f_1(x)$ such that $f_0(x)=f(x)$ for all $x\in \theta+W_{m-1}$ and $f_1(x)=f(x)$ for all $x\in \theta+\beta_{m-1}+W_{m-1}$. The polynomials $f_0(x)$ and $f_1(x)$ can be obtained by taking the remainders of $f(x)$ when divided by the vanishing polynomials of the affine subspaces. Specifically, we have
\begin{equation*}
    \begin{aligned}
        f_0(x) &= f(x) \mod \mathbb{Z}_{W_{m-1}}(x+ \theta) \quad \text{and} \quad
        f_1(x) &= f(x) \mod \mathbb{Z}_{W_{m-1}}(x + \theta + \beta_{m-1}).
    \end{aligned}
\end{equation*}
By applying this reduction step again to $\fft(f_0,\theta + W_{m-1})$ and $\fft(f_1,\theta + \beta_{m-1} + W_{m-1})$, we continue until all the resulting polynomials $f_0(x)$ and $f_1(x)$ are constants.

\begin{algorithm2e}[htb]
    \caption{Von zur Gathen--Gerhard additive FFT of length $n = 2^m$}\label{Algo:Gathen--Gerhard}
    \SetAlgoLined
    { \scriptsize
    \KwIn{$f(x) \in \mathbb{F}_{2^k}[x]$ of degree $< n = 2^m$ and an affine subspace $\theta+W_m=\theta+ \langle \beta_0,\beta_1,\ldots,\beta_{m-1} \rangle$, where $\set{\beta_0=1,\beta_1,\ldots,\beta_{m-1}}$ is any ordered basis of $\mathbb{F}_{2^k}$.}
    \KwOut{$\fft(f, \theta+W_m)$.}

    \If{$m=0$}
    {
        \Return $f(\theta)$.\
    }
    Compute 
    \[
    \begin{aligned}
        f_0(x) &= f(x) \mod \mathbb{Z}_{W_{m-1}}(x + \theta)\ \quad \text{and} \quad
        f_1(x) = f(x) \mod \mathbb{Z}_{W_{m-1}}(x + \theta + \beta_{m-1}).
    \end{aligned}
    \]
    \Return $\fft(f_0,\theta+W_{m-1})||\fft(f_1, \theta+\beta_{m-1}+ W_{m-1})$.\
    }
\end{algorithm2e}

\noindent Cantor's AFFT is obtained by specializing the recursion in Algorithm~\ref{Algo:Gathen--Gerhard} to affine subspaces generated by a Cantor special basis. Under a Cantor special basis, the vanishing polynomials take a particularly simple form. More specifically, if $\{\beta_0 = 1, \beta_1, \ldots, \beta_{m-1}\}$ is a Cantor special basis, then the vanishing polynomial of $W_i = \langle \beta_0, \ldots, \beta_{i-1}\rangle$ is
\(
Z_{W_i}(x) = S^i(x),
\)
i.e. the $i$-fold composition of the mapping $S$ (see Section~\ref{Sec:Preliminaries}). In particular, $Z_{W_i}(x)$ is linearized with all coefficients in $\mathbb{F}_{2}$, so polynomial division by $Z_{W_i}$ requires \emph{zero} $\mathbb{F}_{2^k}$-multiplications.

% \begin{algorithm}[h]
%     \caption{Cantor additive FFT of length $n = 2^m$}\label{Algo:Cantor}
%     \SetAlgoLined
%     { \scriptsize
%     \KwIn{$f(x) \in \mathbb{F}_{2^k}[x]$ of degree $< n = 2^m$, where $k=2^{\ell}$ and the affine subspace $\theta+W_m=\theta+ \langle \beta_0,\beta_1,\ldots,\beta_{m-1} \rangle$, where $\set{\beta_0=1,\beta_1,\ldots,\beta_{m-1}}$ is a Cantor special basis.}
%     \KwOut{$\fft(f, \theta+W_m)$.}

%     \If{$m=0$}
%     {
%         \Return $f(\theta)$.\
%     }
%     Compute 
%     \[
%     \begin{aligned}
%         f_0(x) &= f(x) \mod S^{m-1}(x + \theta),\  \text{ and } \\
%         f_1(x) &= f(x) \mod S^{m-1}(x + \theta + \beta_{m-1}).
%     \end{aligned}
%     \]
%     \Return $\fft(f_0,\theta+W_{m-1})||\fft(f_1, \theta+\beta_{m-1}+ W_{m-1})$.\
%     }
% \end{algorithm}

\end{document}